\documentclass[10pt,journal,amsmath,amssymb,epsfig,psfrag]{IEEEtran}
\title{Interference Alignment for the Multi-Antenna Compound Wiretap Channel}
\usepackage{anysize}
\usepackage[tbtags]{amsmath} % defines many math commands and
\usepackage{amssymb}
\usepackage{verbatim} % get comment environment, + new verbatim
\usepackage{cite}
\usepackage{graphicx}
\usepackage{epsfig}
\usepackage{psfrag}

\newcounter{actr}
{\begin{list}{(\alph{actr})}{\usecounter{actr}}}{\end{list}}

\newcounter{ictr}
{\begin{list}{(\roman{ictr})}{\usecounter{ictr}}}{\end{list}}

\newtheorem{remark}{Remark}
\newtheorem{thm}{Theorem}
\newtheorem{lemma}{Lemma}

\newtheorem{prop}{Proposition}
\newtheorem{defn}{Definition}

\newenvironment{new-proof}[1]
{{\em Proof }:\\}%
{ \noindent\qed }
\newcommand{\qed}{\rule[0.1ex]{1.4ex}{1.6ex}}

\newcommand{\defeq}{\stackrel{\Delta}{=}}

\newcommand{\pe}{\Pr(\eps)}

\newcommand{\mrm}{\mathrm}

\newcommand{\ba}{{\mathbf{a}}}

\newcommand{\cA}{{\mathcal{A}}}

\newcommand{\bb}{{\mathbf{b}}}

\newcommand{\cC}{{\mathcal{C}}}

\newcommand{\cE}{{\mathcal{E}}}

\newcommand{\bg}{{\mathbf{g}}}

\newcommand{\bgt}{{\tilde{\bg}}}

\newcommand{\cG}{{\mathcal{G}}}

\newcommand{\bh}{{\mathbf{h}}}
\newcommand{\bht}{{\tilde{\bh}}}

\newcommand{\cN}{{\mathcal{N}}}

\newcommand{\cS}{{\mathcal{S}}}

\newcommand{\cT}{{\mathcal{T}}}
\newcommand{\bu}{{\mathbf{u}}}

\newcommand{\bv}{{\mathbf{v}}}

\newcommand{\bx}{{\mathbf{x}}}

\newcommand{\xt}{{\tilde{x}}}

\newcommand{\yt}{{\tilde{y}}}

\newcommand{\zt}{{\tilde{z}}}

\newcommand{\al}{\alpha}
\newcommand{\bal}{{\boldsymbol{\al}}}

\newcommand{\bt}{\boldsymbol{t}}

\newcommand{\g}{\gamma}

\newcommand{\del}{\delta}

\newcommand{\eps}{\varepsilon}

\DeclareMathAlphabet{\mathbsf}{OT1}{cmss}{bx}{n}% bold sans serif
\DeclareMathAlphabet{\mathssf}{OT1}{cmss}{m}{sl}% slanted sans serif

\DeclareSymbolFont{bsfletters}{OT1}{cmss}{bx}{n}
\DeclareSymbolFont{ssfletters}{OT1}{cmss}{m}{n}
\DeclareMathSymbol{\bsfGamma}{0}{bsfletters}{'000}
\DeclareMathSymbol{\ssfGamma}{0}{ssfletters}{'000}
\DeclareMathSymbol{\bsfDelta}{0}{bsfletters}{'001}
\DeclareMathSymbol{\ssfDelta}{0}{ssfletters}{'001}
\DeclareMathSymbol{\bsfTheta}{0}{bsfletters}{'002}
\DeclareMathSymbol{\ssfTheta}{0}{ssfletters}{'002}
\DeclareMathSymbol{\bsfLambda}{0}{bsfletters}{'003}
\DeclareMathSymbol{\ssfLambda}{0}{ssfletters}{'003}
\DeclareMathSymbol{\bsfXi}{0}{bsfletters}{'004}
\DeclareMathSymbol{\ssfXi}{0}{ssfletters}{'004}
\DeclareMathSymbol{\bsfPi}{0}{bsfletters}{'005}
\DeclareMathSymbol{\ssfPi}{0}{ssfletters}{'005}
\DeclareMathSymbol{\bsfSigma}{0}{bsfletters}{'006}
\DeclareMathSymbol{\ssfSigma}{0}{ssfletters}{'006}
\DeclareMathSymbol{\bsfUpsilon}{0}{bsfletters}{'007}
\DeclareMathSymbol{\ssfUpsilon}{0}{ssfletters}{'007}
\DeclareMathSymbol{\bsfPhi}{0}{bsfletters}{'010}
\DeclareMathSymbol{\ssfPhi}{0}{ssfletters}{'010}
\DeclareMathSymbol{\bsfPsi}{0}{bsfletters}{'011}
\DeclareMathSymbol{\ssfPsi}{0}{ssfletters}{'011}
\DeclareMathSymbol{\bsfOmega}{0}{bsfletters}{'012}
\DeclareMathSymbol{\ssfOmega}{0}{ssfletters}{'012}

\renewcommand{\pe}{\Pr(e)}
\renewcommand{\defeq}{\triangleq}

\newcommand{\rvW}{{\mathssf{W}}}    % W
\newcommand{\rvm}{{\mathssf{m}}}    % m

\newcommand{\rvbm}{{\mathbsf{m}}}

\newcommand{\rvn}{{\mathssf{n}}}    % n

\newcommand{\rvs}{{\mathssf{s}}}    % s

\newcommand{\rvu}{{\mathssf{u}}}    % u

\newcommand{\rvv}{{\mathssf{v}}}    % v

\newcommand{\rvbv}{{\mathbsf{v}}}

\newcommand{\rvw}{{\mathssf{w}}}    % w

\newcommand{\rvbw}{{\mathbsf{w}}}

\newcommand{\rvx}{{\mathssf{x}}}    % x, random variable

\newcommand{\rvbx}{{\mathbsf{x}}}

\newcommand{\rvy}{{\mathssf{y}}}    % y

\newcommand{\rvyt}{{\tilde{\rvy}}}

\newcommand{\rvby}{{\mathbsf{y}}}

\newcommand{\rvz}{{\mathssf{z}}}    % z

\newcommand{\rvzt}{{\tilde{\rvz}}}

\newcommand{\rvbz}{{\mathbsf{z}}}

\newcommand{\rvV}{{\mathssf{V}}}

\newcommand{\var}{\mrm{var}}

\newcommand{\bbe}{\mathbf{b}}

\newcommand{\cH}{{\mathcal H}}
\newcommand{\dent}{\hbar}

\newtheorem{assump}{Assumption}

\begin{document}

\author{
\authorblockN{Ashish Khisti}\\
\authorblockA{ECE Dept. \\
University of Toronto\\
Toronto, ON, M1B 5P1\\
akhisti@comm.utoronto.ca} \thanks{This work was support by a Natural Science and Engineering Research Council (NSERC) Discovery Grant. Part of the work has been presented at the Information Theory and Applications Worksop (ITA), San Diego, 2010~\cite{khistiITA10}.}}

\maketitle

\begin{abstract}
We study a wiretap channel model where the sender has $M$ transmit antennas and there are two groups consisting of $J_1$ and $J_2$ receivers respectively. Each receiver has  a single antenna. We consider two scenarios. First we consider the compound wiretap model ---  group $1$ constitutes the set of legitimate receivers, all interested in a common message, whereas  group $2$ is the set of eavesdroppers. We establish new lower and upper bounds on the secure degrees of freedom.  Our lower bound is based on the recently proposed \emph{real interference alignment} scheme. The upper bound provides the first known example which illustrates that the \emph{pairwise upper bound} used in earlier works is not tight. 

The second scenario we study is the compound private broadcast channel. Each group is interested in a message that must  be protected from the other group. Upper and lower bounds on the degrees of freedom are developed by  extending the results on the compound wiretap channel. 
\end{abstract}

\section{Introduction}
Wyner's wiretap channel~\cite{wyner:75Wiretap} is an information theoretic model for secure communications at the physical layer. In this model, there are three terminals --- a sender, a receiver and an eavesdropper. A wiretap code simultaneously  meets a reliability constraint with respect to the legitimate receiver and a secrecy constraint with respect to the eavesdropper. In recent times, there has been a significant interest in applying this model to wireless communication systems.  Some recent works include secure communications over fading channels~\cite{khistiTchamWornell:07,gopalaLaiElGamal:06Secrecy,liangPoor07}, multi-antenna wiretap channels~\cite{khistiWornellEldar:07,khistiWornell:09a,khistiWornell:09b,LyLiu:09,LiuLiu:09,LiuShamai:09,shafieeLiuUlukus:07,EkremUlukus:07,Ghadamali} and several multiuser extensions of the wiretap channel. 

The wiretap channel requires that channel statistics of all the terminals be globally known.  This model is justified in applications where  the receiver channels are degraded. The wiretap code can be designed for the strongest (worst-case) eavesdropper in the class of all eavesdropper channels. However in many cases of practical interest, such as in the case of  multi-antenna channels, the receivers cannot be ordered in this fashion. There is no natural choice for the ``worst-case" eavesdropper and the ordering of the eavesdroppers depends on the transmit directions. Hence it is natural to study an extension of the wiretap channel that explicitly incorporates the lack of knowledge of the receiver channels i.e., the compound wiretap channel. This model was recently studied in~\cite{liangKramer:08, Liuprabhakaran:07,khistiTchamWornell:07}.  The channels of the legitimate receiver and  the eavesdropper take one of finitely many values. Note that this problem is equivalent to broadcasting a common message to multiple intended receivers, one corresponding to each channel state, while keeping the message secure against a collection of non-colluding eavesdroppers. A lower bound on the secrecy capacity is established in~\cite{liangKramer:08}. One special case where the optimality of this scheme holds is the deterministic wiretap channel with a single realization of the legitimate receiver. In this case the lower bound coincides with a natural \emph{pairwise upper bound} on the secrecy capacity. The pairwise bound is obtained as follows. We consider the secrecy capacity associated with one particular pair of legitimate receiver and eavesdropper by ignoring the presence of all other terminals. Clearly this constitutes an upper bound on the capacity. The pairwise upper bound is obtained by selecting the pair with the smallest capacity.
The pairwise upper bound was also used in establishing the secrecy capacity in~\cite{khistiTchamWornell:07,Liuprabhakaran:07,gopalaLaiElGamal:06Secrecy} for a class of parallel reversely degraded compound wiretap channels. In~\cite{khistiTchamWornell:07} the authors consider the case of multiple legitimate receivers and one eavesdropper and introduce a new class of \emph{secure multicast codes} that achieve the pairwise upper bound. When specialized to the case of no eavesdroppers, the resulting scheme yields a different coding scheme than the  vector codebook approach in~\cite{elGamal:80}.
The case when there is one legitimate receiver and multiple eavesdroppers is settled in~\cite{Liuprabhakaran:07,gopalaLaiElGamal:06Secrecy}. A new coding scheme  is proposed that meets the pairwise upper bound. Some other recent works on the compound wiretap channel include~\cite{EkremUlukus:09,PeronDiggavi:09}.

To the best of our knowledge, no upper bounds, besides the pairwise upper bound, are known for the compound wiretap channel. In this paper we study the multi-input-single-output (MISO) wiretap channel, where both the legitimate receivers and the eavesdroppers channel take one of finitely many states. We develop a new upper bound on secrecy-rate that is tighter than the pairwise upper bound and establishes that in general there is a loss in degrees of freedom due to uncertainty of channel state information at the transmitter. In addition we develop new lower bounds that  combine the real interference alignment technique recently proposed in~\cite{realInterf1,realInterf2,compoundMIMO1,compoundMIMO2} with wiretap code constructions.  Our achievable degrees of freedom remain constant, independent of the number of states of the legitimate receiver and eavesdropper. In contrast we observe that naive approaches based on time-sharing only achieve vanishing degrees of freedom as the number of states increase.

We also study an extension of the compound MISO wiretap channel when there are two messages, that we refer to as the \emph{compound private broadcast}. To our knowledge the private broadcast model  is first proposed by Cai and Lam~\cite{caiLam00}. While~\cite{caiLam00} only studies the deterministic broadcast channel,  more recent works~\cite{LiuMaricSpasojevicYates:07,LiuLiu:09} study a larger class of channels including the discrete memoryless channels and the multi-input-multi-output Gaussian channels. The present paper extends this model to the case when each receiver's channel takes one of finitely many states. Lower and upper bounds on the sum of the secure degrees of freedom are developed. While we restrict our analysis to the above mentioned cases,  we expect similar techniques to be applicable to other extensions of the wiretap channel such as~\cite{IA:08}\cite{heYener:09}\cite{EkremUlukus:07}\cite{Ghadamali}.

The remainder of the paper is organized as follows. Section~\ref{sec:Main} described the channel model and summarizes the main results in this is paper. In section~\ref{sec:realInterf} we review the real interference alignment scheme for the scalar point-to-point Gaussian channel. Sections~\ref{sec:CWC:LB} and~\ref{sec:CWC:UB}  establish lower and upper bound on the secrecy degrees of freedom of the compound wiretap channel. Sections~\ref{sec:PBC:LB} and~\ref{sec:PBC:UB} develop analogous results for the compound private broadcast channel. Conclusions are provided in section~\ref{sec:concl}.

\section{Main Results}
\label{sec:Main}
The channel model consists of one transmitter with $M$ antennas and two receivers each 
with one antenna. We further assume that the channels coefficient vectors of the two receivers belong of a finite set i.e.,
\begin{equation}
\begin{aligned}
\bh &\in \cH =\left\{\bh_1,\bh_2,\ldots,\bh_\mrm{J_1}\right\}\\
\bg &\in \cG =\left\{\bg_{1},\bg_{2},\ldots,\bg_{J_2}\right\}\\
\end{aligned}
\end{equation}
It is assumed that each receiver knows its own channel realization whereas the remaining terminals are only aware of the sets $\cH$ and $\cG$. 
Furthermore we assume that the channel coefficients remain fixed for the entire duration of communication. In our analysis of lower and upper bounds we make one of the following two assumptions.
\begin{assump}
\label{assump:LB}
The channel vectors $\bh_1,\ldots, \bh_\mrm{J_1}$ as well as $\bg_1,\ldots, \bg_{J_2}$ are each drawn from a real valued continuous distribution. 
\end{assump}
\begin{assump}
\label{assump:UB}
Any collection of $M$ (or fewer vectors) in $\cH \cup \cG$ be linearly independent.
\end{assump}
We note that assumption~\ref{assump:LB}, almost surely implies assumption~\ref{assump:UB} and in this sense it is stronger.  The first assumption is used in the analysis of the lower bound whereas the second assumption is used in the analysis of the upper bound. 

The resulting channel model can be expressed as
\begin{equation}
\begin{aligned}
\rvy_j &= \bh_j^T \rvbx + \rvv_j,\qquad j=1,\ldots,J_1\\
\rvz_k &= \bg_k^T \rvbx + \rvw_k,\qquad k=1,\ldots,J_2 \label{eq:model1}
\end{aligned}
\end{equation}
where the channel input vector $\rvbx$ is required to satisfy the average power constraint $E[||\rvbx||^2]\le P$, the additive noise variables $\rvv_j$ and $\rvw_k$ are i.i.d.\ and distributed $\cN(0,1)$.

In the remainder of this section we separately consider two cases: the compound wiretap channel and the compound private broadcast channel.

\subsection{Compound Wiretap Channel}
A compound wiretap encoder  maps a message $\rvm$, uniformly distributed over a set of size $2^{nR}$, to the channel input sequence $\bx^n$. The decoder produces a message estimate $\hat{m}_j=g_j(\rvy_j^n;\bh_j)$. A rate $R$ is achievable if there exist a sequence of encoder and decoders of such that $\pe = \Pr(\rvm\neq \hat{\rvm}_j)\rightarrow 0$ as $n\rightarrow \infty$ for each $j=1,2,\ldots, J_1$ and $\frac{1}{n}I(\rvm;z_j^n) \rightarrow 0$ for each $j=1,2,\ldots, J_2$. The largest rate achievable under these constraints is the \emph{compound secrecy capacity}. Of particular interest in this paper is the degrees of freedom (d.o.f.) of the compound wiretap channel. We say that $d$ is an achievable secure degrees of freedom for the compound wiretap channel, if there exists a sequence of achievable rates $R(P)$, indexed by power $P$, such that\begin{equation}d = \lim_{P\rightarrow \infty}\frac{R(P)}{\frac{1}{2}\log_2 P}.\end{equation}The maximum attainable value of $d$ is the \emph{secrecy d.o.f.} of the compound wiretap channel.
\begin{figure*}
\includegraphics[scale=0.65]{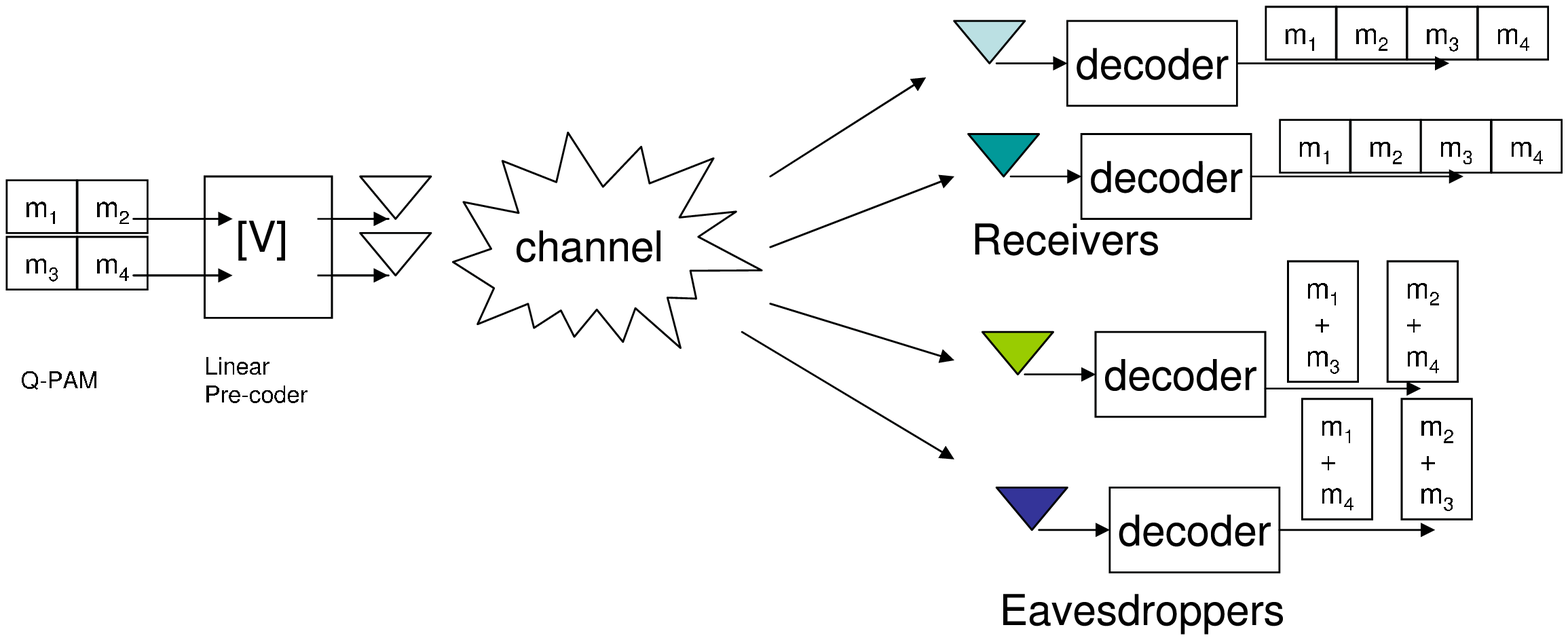}
\label{fig:CWC:IA}

\caption{Interference Alignment for the compound wiretap channel. Each of the four messages is drawn from a PAM constellation and carries a rate $\approx \frac{1}{8}\log P$. The linear precoder $V$ guarantees that while each legitimate receiver can decode all the four messages, each eavesdropper can only obtain two integer linear combinations of the messages (and no other information about the messages), thus reducing its signal dimension by a factor of $2$. Precoding matrices based on real interference alignment techniques~\cite{realInterf1,realInterf2} enable us to reduce the signal dimension at an arbitrary number of eavesdroppers by a factor of $\approx \frac{1}{M}$, thereby achieving $\approx 1-\frac{1}{M}$ degrees of freedom in Theorem~\ref{thm:CWC:LB}.}
\end{figure*}

We develop the following lower and upper bounds on the secure degrees of freedom.
\begin{thm}
\label{thm:CWC:LB}
Under assumption~\ref{assump:LB}, the following secure degrees of freedom are achievable  for the compound wiretap channel for all channel coefficient vectors, except a set of measure zero,
\begin{equation}
d_L = \begin{cases}
1, & \min(J_1,J_2)  < M,\\
\frac{M-1}{M}, & \min(J_1,J_2) \ge M.
\end{cases}
\label{eq:CWC:LB}
\end{equation}\hfill\qed
\end{thm}
The lower bound, for the case $\min(J_1,J_2)\ge M$, is achieved by combining  real-interference alignment with a wiretap code construction. We can interpret the resulting degrees of freedom as follows.  By using interference alignment, the transmitter chooses signalling dimensions such that at each eavesdropper, the received signal dimensions are reduced by a factor of approximately $\frac{1}{M}$, whereas each intended receiver incurs no loss in the received signal dimensions. A wiretap code can then be designed to take advantage of this discrepancy to achieve $1-\frac{1}{M}$ degrees of freedom. 

It is noteworthy that interference alignment significantly outperforms a naive time-sharing based scheme where the achievable degrees of freedom approach zero as the number of states becomes large.
\begin{prop}
A scheme that combines time-sharing and noise transmission achieves the following degrees of feedom
\begin{equation}
d_L^{TS} = \begin{cases}
1, & \min(J_1,J_2) < M \\
\frac{M-1}{\min(J_1,J_2)}, & \min(J_1,J_2)\ge M \label{eq:CWC:TSLB}
\end{cases}
\end{equation}\label{prop:CWC:TSLB}
\hfill\qed\end{prop}
Comparing~\eqref{eq:CWC:LB} and~\eqref{eq:CWC:TSLB} it is clear that interference alignment provides significant gains in the degrees of freedom compared to time-sharing based lower bounds. The following example considers the case of rationally dependent channel gains and shows that higher achievable degrees of freedom can be achieved by a multilevel coding scheme. 
\begin{prop}
\label{prop:CWC:ML}
Consider a special case of channel model~\eqref{eq:model1}, where $M=J_1 = J_e =2$ and furthermore let
\begin{equation}\begin{aligned}
&\rvy_1 = \rvx_1 + \rvx_2 + \rvv_1, \qquad \rvy_2 = \rvx_1 - \rvx_2 + \rvv_2\\
&\rvz_k = \rvx_k + \rvw_k,\qquad k=1,2
\end{aligned}\label{eq:CWC:ML:Model}\end{equation} corresponding to the choice of $\bh_1 = [1;1]$, $\bh_2 = [1;-1]$, $\bg_1 = [1;0]$ and $\bg_2 = [0;1]$. There exists a multi-level coding scheme that achieves $\log_3 2 \approx 0.63$ secure d.o.f.\
\hfill\qed\end{prop}
As will become apparent, the proposed multi-level coding scheme bears similarity with the interference alignment technique in that both schemes force the eavesdropper receivers to decode  linear combination of transmitted symbols. However while the interference alignment technique requires each legitimate receiver to decode every message symbol, the proposed multi-level coding scheme relaxes this constraint by taking advantage of the special channel structure in~\eqref{eq:CWC:ML:Model}.

The proof of  Theorem~\ref{thm:CWC:LB} and Prop.~\ref{prop:CWC:ML} are provided in section~\ref{sec:CWC:LB}. The proof of Prop.~\ref{prop:CWC:TSLB} appears in Appendix~\ref{app:CWC:TSLB}.
\begin{thm}
Under assumption~\ref{assump:UB}, the following expression provides an upper bound on the secure d.o.f.\ of the compound wiretap channel 
\begin{equation}
d_U = \begin{cases}1, &\min(J_1,J_2) < M\\
1-\frac{1}{M^2-M+1}, & \min(J_1, J_2) \ge M \end{cases}\label{eq:CWC:UB}
\vspace{-0.5em}\end{equation}
\label{thm:CWC:UB}\hfill\qed\end{thm}

A new upper bound is derived in the proof of Theorem~\ref{thm:CWC:UB} by considering the constraints imposed due to secrecy and common message transmission. As we show below,  the single-letter upper bounds in earlier works only yield 1 d.o.f.\ In particular the pairwise upper bound on the secrecy capacity is 
\begin{equation}
C \le \max_{p_\rvbx} \min_{j,k} I(\rvbx;\rvy_j|\rvz_k)\label{eq:pairwise}
\end{equation}
This bound can be interpreted as follows: consider receiver $j$ and eavesdropper $k$. An upper bound on the secrecy capacity, in absence of all other terminals, is $I(\rvbx;\rvy_j|\rvz_k)$.  Minimizing over all such pairs results in~\eqref{eq:pairwise}. To the best of our knowledge this upper bound has been shown to be tight in some special cases~\cite{khistiTchamWornell:07,Liuprabhakaran:07,liangKramer:08}. For the compound MISO wiretap channel however this upper bound results in 1 secure d.o.f.  Indeed with $\bx \sim \cN(0,\frac{P}{M}I)$
\begin{align*}
&I(\rvbx;\rvy_i|\rvz_j) = I(\rvbx; \rvy_i,\rvz_j) - I(\rvbx; \rvz_j)\\
&= \sum_{j=1}^2 \frac{1}{2}\log\left(1+\frac{\lambda_j(H) P}{M}\right)- \frac{1}{2}\log\left(1 + \frac{P}{M}||\bg_j||^2\right)
\end{align*}
where $\lambda_1(H)$ and $\lambda_2(H)$ are the two non-zero eigen-values of the matrix $H =\left(
\begin{array}{c}
\bh_i^T  \\
\bg_j^T
\end{array}
\right)
\left(
\begin{array}{cc}
\bh_i &
\bg_j
\end{array}
\right)$. This yields 1 d.o.f.

To improve the pairwise upper bound in~\eqref{eq:pairwise} one can incorporate the fact that each receiver wants a common message i.e., 
\begin{equation}
C \le \max_{p_\rvbx} \min_{i,j} \min\{I(\rvbx;\rvy_i|\rvz_j), I(\rvbx;\rvy_i)\}\label{eq:pairwise2}
\end{equation}
however it can easily be verified that this potentially tighter bound also yields 1 d.o.f.

Instead of applying the single-letter bounds above  in Theorem~\ref{thm:CWC:UB} we start with the multi-letter characterization and carefully combine the associated constraints to get the proposed upper bound. We sketch the main steps below.

First, via the secrecy constraint, we show that
\begin{equation}
nR \le \dent(\rvby_1,\ldots, \rvby_M) - \frac{1}{M}\dent(\rvbz_1,\ldots,\rvbz_M) + n\eps_n, \label{eq:cond1}
\end{equation}
where one can interpret the  first term as the received signal at an enhanced user who observed all $(\rvby_1,\ldots, \rvby_M)$
whereas the second term corresponds to the observation at an ``average eavesdropper". Thus for the rate to be large, we need either the joint entropy of the eavesdropper observations $(\rvbz_1,\ldots,\rvbz_M)$ to be small or the joint entropy of the legitimate receivers $(\rvby_1,\ldots,\rvby_M)$ to be large.

Next we show that the joint entropy of  $(\rvbz_1,\ldots,\rvbz_M)$ cannot be too small compared to the joint entropy of $(\rvby_1,\ldots, \rvby_M)$.  Recall that with $J_1 \ge M$ and $J_2 \ge M$,  both $(\bh_1,\ldots, \bh_M)$  and $(\bg_1,\ldots,\bg_M)$ constitute a basis of $\mathbb{R}^M$. Using this property we show that
\begin{equation}
\dent(\rvbz_1,\ldots,\rvbz_M) \ge \dent(\rvby_1,\ldots,\rvby_M) - nMd,\label{eq:cond2}
\end{equation}
where $d$ is a constant that does not depend on $P$. Combining~\eqref{eq:cond1} and~\eqref{eq:cond2} we can deduce that
\begin{equation}
nR \le \left(1-\frac{1}{M}\right)\dent(\rvby_1,\ldots,\rvby_M) + nd_1\label{eq:cond3}
\end{equation}where $d_1 = d+\eps_n$. 

It thus follows that for the rate $R$ to be large, we  the joint entropy of $(\rvby_1,\ldots,\rvby_M)$ must be large. However since a common message needs to be delivered to each of the receivers, the outputs need to be sufficiently correlated. In particular we show that
\begin{equation}
n(M-1)R \le \sum_{i=1}^M\dent(\rvby_i) - \dent(\rvby_1,\ldots, \rvby_M) + nM\eps_n.\label{eq:cond4}
\end{equation}Eq.~\eqref{eq:cond3} and~\eqref{eq:cond4} illustrate the tension between the secrecy and common message transmission constraint. For~\eqref{eq:cond3} to be large we need the output sequences $(\rvby_1,\ldots, \rvby_M)$ to be as independent as possible.  However the common message constraint penalizes such independence.  Our upper bound in Theorem~\ref{thm:CWC:UB} exploits this tension. A complete proof appears in section~\ref{sec:CWC:UB}.

\subsection{Compound Private Broadcast}

An encoder for the compound private broadcast channel maps a message pair $(\rvm_1,\rvm_2)$, distributed uniformly and independently over sets of size $2^{nR_1}$ and $2^{nR_2}$ respectively, to the channel input sequence $\bx^n$. The decoders in group $1$ produces a message estimate $\hat{\rvm}_{1j}=g_{1j}(\rvy_j^n;\bh_j)$ while the decoders in group $2$ produce a message estimate $\hat{\rvm}_{2k} = g_{2k}(\rvz_k^n;\bg_k)$. A rate pair $(R_1,R_2)$ is achievable if there exist a sequence of encoder and decoders of such that $\pe = \Pr(\{\rvm_1\neq \hat{\rvm}_{1j}\}_{j=1}^{J_1} \cup \{\rvm_2 \neq \hat{\rvm}_{2k}\}_{k=1}^{J_2})\rightarrow 0$ as $n\rightarrow \infty$ and $\frac{1}{n}I(\rvm_1;\rvz_k^n) \rightarrow 0$ and $\frac{1}{n}I(\rvm_2;\rvy_j^n)\rightarrow 0$ for each $k=1,\ldots, J_2$ and $j=1,\ldots, J_1$. The set of all achievable rate pairs under these constraints constitutes the capacity region. 

Of particular interest is the sum secrecy degrees of freedom (d.o.f.). We say that $d^s$ is  achievable if there exists a sequence of achievable rate pairs $(R_1(P),R_2(P))$, indexed by power $P$, such that\begin{equation}d^s = \lim_{P\rightarrow \infty}\frac{R_1(P)+R_2(P)}{\frac{1}{2}\log_2 P}.\end{equation} The maximum attainable value of $d^s$ is the \emph{sum-secrecy d.o.f.} of the compound private broadcast channel.

We develop the following lower and upper bounds on the sum-secrecy degrees of freedom. 
\begin{thm}
\label{thm:CPB:LB}
Under assumption~\ref{assump:LB}, almost surely, the following sum-secrecy degrees of freedom are achievable for the compound private broadcast channel
\begin{equation}\begin{aligned}
d^s_L = \begin{cases}2, & \max(J_1, J_2) < M \\
2\frac{M-1}{M},  & \max(J_1, J_2) \ge M > \min(J_1, J_2) \\
2\frac{M-1}{M+1}, &\min(J_1,J_2) \ge M \end{cases}
\end{aligned}\label{eq:PBC:LB}\end{equation}
\end{thm}
The coding scheme, presented in section~\ref{sec:PBC:LB}, also combines wiretap codes with the interference alignment scheme. 
The following theorem provides an upper bound on the sum secrecy degrees of freedom 
\begin{thm}
\label{thm:CPB:UB}
Under assumption~\ref{assump:UB}, an upper bound on the sum secrecy degrees of freedom of the compound private broadcast channel is
\begin{equation}
\label{eq:PBC:UB}
d_U^s = \begin{cases}2, &\max(J_1, J_2) < M \\ 
\frac{2M-1}{M}, &\min(J_1, J_2) < M \le \max(J_1,J_2) \\\ 2\frac{M-1}{M}, &\min(J_1, J_2)\ge M.\end{cases}
\end{equation}
\end{thm}
A proof provided in section~\ref{sec:PBC:UB} extends the techniques in the proof of Theorem~\ref{thm:CWC:UB}.

\begin{remark}
Throughout this paper we assume the channels to be real valued. However we do not expect the results to be different for complex valued coefficients. In particular, the upper bounds in Theorem~\ref{thm:CWC:UB} and~\ref{thm:CPB:UB} immediately extend to the complex channel coefficients as they are developed using standard techniques. The lower bounds in Theorem~\ref{thm:CWC:LB} and~\ref{thm:CPB:LB} are based on the real interference alignment scheme. Using its recent extension to complex channel coefficients sketched in~\cite{compoundMIMO1}, we expect similar results to hold for  complex valued channel coefficients  as well. 
\end{remark}

\section{Real Interference Alignment}
\label{sec:realInterf}
In this section we review the main results  of real interference alignment from~\cite{realInterf1,realInterf2}.  For simplicity, we describe this scheme for a point-to-point scalar channel. 
\begin{equation}
\rvy = \rvx + \rvz,
\end{equation}
where the input satisfies a power constraint $E[\rvx^2]\le P$ and the additive noise $\rvz \sim \cN(0,\sigma^2)$ is Gaussian. Assume that the input symbols are drawn from a PAM constellation,\begin{equation}
\cC_0 = a_0 \left\{-Q_0,-Q_0+1,\ldots, Q_0-1, Q_0\right\}.\label{eq:PAM0}
\end{equation}Two quantities associated with this constellation are the minimum distance and the rate. In particular $a_0=d_\mrm{min}(\cC_0)$  governs the error probability according to the relation
\begin{equation}
\pe \le \exp\left(-\frac{a_0^2}{8\sigma^2}\right)\label{eq:pe0},
\end{equation}
while the rate is given by
\begin{equation}\label{eq:r0}
R =\frac{1}{2}\log(1+2Q_0).
\end{equation}
Furthermore the choice of $Q_0$ and $a_0$ must satisfy the average power constraint
\begin{equation}
E[x^2] = \frac{Q_0^2 a_0^2}{12} \le P.   
\end{equation} 
For an arbitrary $\eps > 0$, select, $Q_0 = P^{\frac{1-\eps}{2}} $ and $a_0 =P^{\frac{\eps}{2}} $. Then,
\begin{equation}
\pe \le \exp\left\{-\frac{P^\eps}{\sigma^2}\right\}, \qquad R \approx \frac{1-\eps}{2}\log P
\end{equation} shows that the error probability can be made sufficiently small by selecting $P$ large enough and furthermore the rate is close to the Shannon limit.

The idea behind real interference alignment is to  have multiple PAM constellation symbols instead of a single constellation~\eqref{eq:PAM0} and thus convert the channel into a multi-input-single-output channel. In particular consider a constellation
\begin{equation}
\cC = a \left\{-Q,-Q+1,\ldots, Q-1, Q\right\}.\label{eq:PAMc}
\end{equation}
and suppose a total of $L$ points $b_1,\ldots, b_L$ are drawn independently from this constellation. The transmit vector is of the form
\begin{equation}
\rvx = \bal^T\bb = \left[\al_1,\ldots, \al_L\right]\left[\begin{array}{c}b_1\\\vdots\\b_L\end{array}\right]\label{eq:RI-Tx}
\end{equation}
where   $\al_1,\ldots, \al_L$ are rationally independent constants (see def.~\ref{def:RI}). 

\begin{figure*}
\psfrag{b1}{$b_1$}\psfrag{b2}{$b_2$}\psfrag{bM}{$b_L$}\psfrag{a1}{$\al_1$}\psfrag{a2}{$\al_2$}\psfrag{aM}{$\al_L$}
\psfrag{ctx}{$\cC$}\psfrag{crx}{$\cC_r(\bal)$}\psfrag{d}{$d_\mrm{min}$}
\includegraphics[scale=0.4]{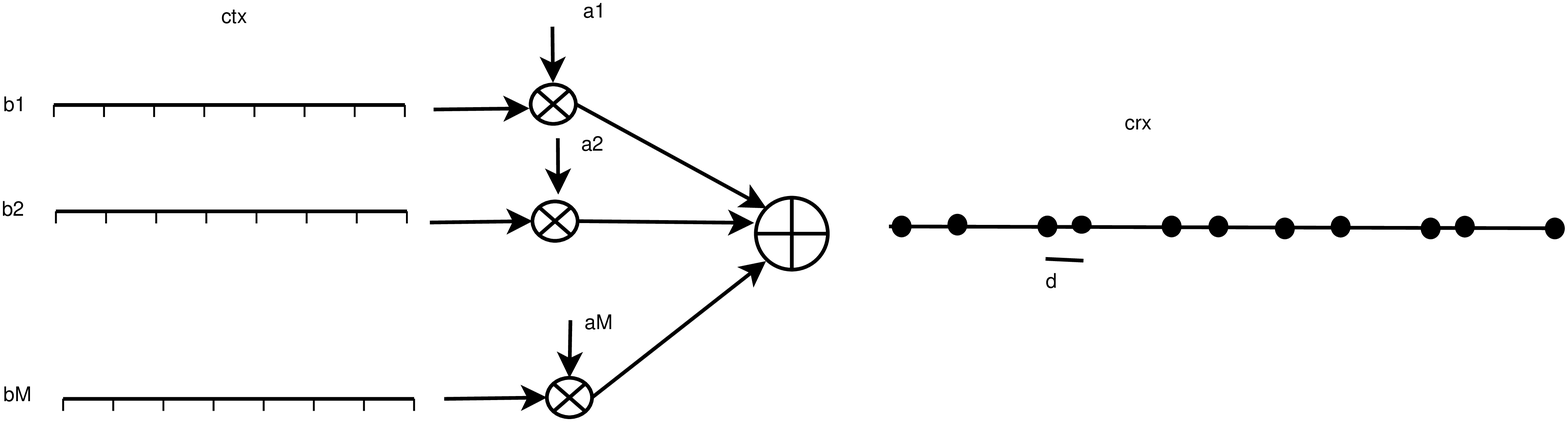}
\caption{Constellations for the real interference alignment scheme. The transmitter constellations are shown on the left hand side. Each of the $M$ points $b_1,\ldots, b_M$ are sampled independently from $\cC = a\{\-Q, \ldots, Q\}$. The receiver constellation $\cC_r(\bal)$ consists of $(2Q+1)^M$ of the form $\rvx  = \bal^T\bb$. The minimum distance in $\cC_r(\bal)$ is stated in Lemma~\ref{lem:dmin}.}\label{fig:realInter}
\end{figure*}
As shown in Fig.~\ref{fig:realInter}, while the transmit constellation is given by $\cC$ in~\eqref{eq:PAMc} and consists of $(2Q+1)$ points, the receiver constellation $\cC_r(\bal)$ consists of all $(2Q+1)^M$ points specified in~\eqref{eq:RI-Tx} i.e.,
\begin{equation}
\cC_r(\bal) =\left\{\rvx \in \mathbb{R} \bigg| \exists \bb \in \cC^M, \rvx = \bal^T\bb  \right\}\label{eq:PAMr}.
\end{equation}
In our subsequent discussion we drop the explicit dependence of $\cC_r$ on $\bal$. 

One key result in~\cite{realInterf1,realInterf2} is a minimum distance between points in $\cC_r$  that depends on the constellation parameters $Q$, $a$ and $L$ and holds for  all vectors $\bal$, except a set of measure 0.
\begin{lemma}[\cite{realInterf1,realInterf2}]
For any $\eps > 0$ there exists a constant $k_\eps$ such that 
\begin{equation}
d_\mrm{min}(\cC_r) \ge \frac{k_\eps a}{Q^{L-1+\eps}} \label{eq:dmin},
\end{equation}
for all vectors $\bal$, except a set of measure zero.
\label{lem:dmin}\end{lemma}

The second key observation in~\cite{realInterf1,realInterf2} is that there is a one-to-one mapping between $\bx$ and $\bb$ in~\eqref{eq:RI-Tx} if $\bal$ consists of rationally independent coefficients. 
\begin{defn}
We say that $\al_1,\ldots, \al_L$ are rationally independent real numbers if the equation $\sum_{i=1}^L \al_i c_i = 0$ has no  solution in $(c_1,\ldots, c_L)$ involving only rational numbers. 
\label{def:RI}
\end{defn}
Provided the vector $\bal$ consists of only rationally indpendent numbers, given an element $\bx \in \cC_r$, the decoder can uniquely identify the vector of message symbols $\bb$. Hence the error probability is given by
\begin{equation}
\pe  \le \exp\left\{-\frac{d^2_\mrm{min}(\cC_r)}{8\sigma^2}\right\} = \exp\left\{-\frac{k_\eps^2 a^2}{ 8\sigma^2Q^{2(L-1+\eps)}}\right\}\label{eq:pe:RI}
\end{equation}
where we have used the expression for $d_\mrm{min}$ stated in Lemma~\ref{lem:dmin}.
Finally with an appropriate choice of $Q$ and $a$, one can approach the Shannon limit while guaranteeing an arbitrarily small error probability in the high SNR regime.
\begin{prop}{\cite{realInterf1,realInterf2}}
Suppose that we select
\begin{equation}
Q =  P^{\frac{1-\eps}{2(L+\eps)}},\qquad a = \frac{P^{1/2}}{||\bal||Q},\label{eq:param:RI}
\end{equation}
then we have that $E[\rvx^2] \le P$ and furthermore for all values of $\bal$, except a set of measure zero, we have that
\begin{equation}
\pe \le \exp\left\{-\eta P^\eps\right\},
\end{equation}where $\eta>0 $ is a constant that depends on the channel coefficients $\bal$, but does not depend on $P$.
\label{prop:paramSelec}\end{prop}
\begin{proof}
To show that the power constraint is satisfied, note that since the elements of $\bb$ are selected independently from $\cC$, we have that
\begin{align}
E[\rvx^2] &= E\left[\left(\sum_{i=1}^L \al_i b_i\right)^2\right]\\
&= \sum_{i=1}^L \al_i^2 E[b_i^2]\\
&\le a^2 Q^2 ||\bal||^2.
\end{align}Thus the choice of $a$ in~\eqref{eq:param:RI} guarantees that $E[\rvx^2] \le P$.

Substituting the value of $a$ in~\eqref{eq:param:RI} into~\eqref{eq:pe:RI} we have that for all $\bal$, except a set of measure zero,
\begin{align}
\pe &\le  \exp\left\{-\frac{k_\eps^2 P}{ 8\sigma^2Q^{2(L-1+\eps)} Q^2 ||\bal||^2}\right\} \notag\\
&= \exp\left\{-\frac{k_\eps^2 P}{ 8\sigma^2Q^{2(L+\eps)} ||\bal||^2}\right\}\notag\\
&=\exp\left\{-\frac{k_\eps^2 P^\eps}{ 8\sigma^2 ||\bal||^2}\right\}\label{eq:subQ}\\
&=\exp\{-\eta P^\eps\}
\end{align}
where we have used the relation~\eqref{eq:param:RI} in~\eqref{eq:subQ} and 
$$\eta = \frac{k_\eps^2}{8\sigma^2||\bal||^2}>0 $$ is a constant that does not depend on $P$.  
\end{proof}

Note that the overall rate that one can achieve with this multiplexed code,
\begin{equation}
R = L \log(2Q+1) \ge \frac{L(1-\eps)}{L+\eps}\frac{1}{2}\log P,
\end{equation}
 can be made arbitrarily close to $\frac{1}{2}\log P$ by selecting $\eps$ to be sufficiently small.  

While the approach of multiple constellation points does not provide any gains over using a single PAM constellation~\eqref{eq:PAM0} in the point-to-point case, the flexibility in choosing any vector $\bal$ consisting of rationally independent elements, has been used
in~\cite{realInterf1,realInterf2} to create interference alignment for the $K-$ user interference channel and in~\cite{compoundMIMO1, compoundMIMO2} for the compound broadcast channel. In this work we show that this approach can also provide significant gains for the compound wiretap channel.

\section{Compound Wiretap Channel: Lower Bounds}
\label{sec:CWC:LB}
In this section we develop the lower bounds on the secure degrees of freedom for the compound wiretap channel.

In subsections~\ref{subsec:CWC:LB:T1a} and~\ref{subsec:CWC:LB:T1b} we provide the proof of Theorem~\ref{thm:CWC:LB} for the case when $\min(J_1,J_2)<M$
and $\min(J_1,J_2)\ge M$ respectively. Subsection~\ref{subsec:CWC:LB:PML} provides a proof of Prop.~\ref{prop:CWC:ML}.

For the proof of Theorem~\ref{thm:CWC:LB}  our approach is to evaluate the following lower bound on the secrecy capacity for a specific  input distribution.
\begin{prop}{~\cite{liangKramer:08}}
An achievable secrecy rate for the compound wiretap channel model~\eqref{eq:model1} is
\begin{align}
R &= \max_{p_{\rvu,\rvx}}\left\{\min_j I(\rvu;\rvy_j) - \max_k I(\rvu;\rvz_k) \label{eq:Rdiff}\right\}
\end{align}
for a choice of random variables $(\rvu,\rvx)$ such that $\rvu \rightarrow \rvx \rightarrow (\rvy_i, \rvz_k)$ is satisfied.
\end{prop}

\subsection{Proof of Theorem~\ref{thm:CWC:LB}: $\min(J_1,J_2) < M$}
\label{subsec:CWC:LB:T1a}
When either $J_1 < M$ or $J_2 < M$ we achieve full 1 d.o.f. through a combination of zero-forcing and noise transmission techniques. 

When $J_2 < M$ i.e., when the number of eavesdropper states is less than $M$, we zero-force all the eavesdroppers and achieve a rate that scales as $\frac{1}{2}\log P$. In particular note that the matrix $G = [\bg_1, \bg_2,\ldots, \bg_{J_2}] \in \mathbb{R}^{M \times J_2}$ has a rank equal to $J_2 < M$. Construct a matrix $B \in \mathbb{R}^{M-J_2 \times M}$, with orthogonal rows, such that $B\cdot G = 0$.
Furthermore, almost surely, each $\bh_j$ is linearly independent of the columns of $G$ and hence $B \bh_j \neq 0$ for $j=1,\ldots, J_2$.
The transmitted vector is
\begin{equation}
\bx = B^T \rvbm, \label{eq:tzx}
\end{equation}
where the information vector $\rvbm \sim \cN\left(0,\frac{P}{M } {\bf I}\right)$ is a vector of i.i.d\ Gaussian symbols. Since any information transmitted along rows of $B$ will not be seen by any eavesdropper, setting $\rvu = \rvbm$ and $\rvbx$ in~\eqref{eq:tzx}, we have that 
\begin{align}
R &= \min_{j} I(\rvbm; y_j)\\
&= \min_{j}\frac{1}{2}\log\left(1 + {\frac{P}{M}||B\bh_j||^2}\right)
\end{align}
which scales as $\frac{1}{2}\log P$ as $B \bh_j \neq 0$ for each $j=1,2,\ldots, J_1$.

Similarly when $J_1 < M$, we achieve 1 d.o.f. by transmitting a noise signal in the common null space of the channel of legitimate receivers. While each intended receiver only observes a clean signal, each eavesdropper receives a superposition of signal and noise and its rate does not increase unboundedly with $P$.

In particular, the matrix $H=[\bh_1,\ldots, \bh_{J_1}]$ has a rank $J_1 < M$. Construct  $A \in {\mathbb{R}}^{M-J_1 \times M}$ with orthogonal rows that satisfy $A\cdot H=0$. Furthermore, almost surely,  $\bg_j$ is linearly independent of the columns of  $H$, and hence $A \bg_j \neq 0$. 
Let the transmitted vector be,
\begin{equation}
\bx = \bt \rvs + A^T {\mathbf n}, \label{eq:LowUsers}
\end{equation}
where $\bt$ is any unit norm vector such that $\bh_j^T\bt \neq 0$ for $j=1,2,\ldots, J_1$, and $\rvs \sim \cN(0,P_0)$ is the information bearing symbol, and ${\bf n} \sim \cN(0,P_0 {\bf I}_{M-J_1})$ is a vector of noise symbols transmitted in the common null-space of user matrices and where $P_0 = \frac{P}{M}$ is selected to meet the transmit power constraint. Accordingly the received signals can be expressed as, \begin{align}
\rvy_j &= \bh_j^T\bx + \rvv_j \\
&= \bh_j^T \bt \rvs + \rvv_j,
\end{align}
and
\begin{align}
\rvz_k &= \bg_k^T \bx + \rvw_k,\\
&= \bg_k^T\bt \rvs + \bg_k^T A^T {\mathbf n} + \rvw_k.
\end{align}
An achievable secrecy rate with $\rvu = \rvs$ and $\rvbx$ in~\eqref{eq:LowUsers} is
\begin{align}
R&= \min_{j} I(\rvs; \rvy_j)-\max_{k} I(\rvs; \rvz_k) \\
&= \min_{j}\frac{1}{2}\log(1 + P_0 |\bh_j^T \bt|^2) \notag\\&\qquad - \max_{k}\frac{1}{2}\log\left(1 + \frac{P_0 |\bg_k^T\bt|^2}{1+ P_0 ||A\bg_k||^2}\right),\label{eq:TsRate}
\end{align}
which scales like $\frac{1}{2}\log P$ since $||A\bg_k|| > 0$ for each $k=1,2,\ldots, J_2$ and $\bh_j^T \bt> 0$ for $j=1,\ldots, J_1$.

\subsection{Proof of Theorem~\ref{thm:CWC:LB}: $\min(J_1,J_2) \ge M$}
\label{subsec:CWC:LB:T1b}

To establish an achievable rate, we again evaluate~\eqref{eq:Rdiff}  for a certain choice of input distributed given by the real interference alignment scheme in~\cite{realInterf1,realInterf2}.

To describe this choice, we begin by defining a set\begin{equation}\label{eq:RI-Ct}
\cT = \left\{\prod_{k=1}^{J_2}\prod_{i=1}^M g_{ki}^{\alpha_{ki}} \big|~  \al_{ki} \in \{0,\ldots, N-1\}\right\},
\end{equation}where $g_{ki}$ denotes the channel gain between the $i-$th antenna and the $k-$th eavesdropper. Each selection of the tuple $\{\al_{ki}\} \in \{0,\ldots, N-1\}^{J_2M}$ results in a different element of $\cT$ and there are a total of $L = N^{MJ_2}$ elements in this set. In addition let,
\begin{equation}
V = \left[\begin{array}{cccc} \bv^T & 0 &\cdots & 0 \\ 0 & \bv^T & \cdots & 0 \\ \vdots &   & \ddots & \vdots \\ 0 & 0 &\cdots & \bv^T \end{array}\right]\in \mathbb{R}^{M \times ML}
\label{eq:RI-V}\end{equation}
be a matrix where $\bv \in \mathbb{R}^{L}$ consists of all elements in the set $\cT$. Furthermore let $\bb\in \mathbb{R}^{ML \times 1}$ be a vector whose entries are sampled independently and uniformly from a PAM constellation
\begin{equation}
\label{eq:C-PAM}
\cC= a\left\{-Q, -Q+1,\ldots, Q-1, Q\right\}.
\end{equation}Our choice of parameters in~\eqref{eq:Rdiff} is
\begin{equation}
\rvbx = V\bb, \qquad \rvu = \bb.\label{eq:params}
\end{equation}

We next state two lemmas, helpful in evaluating~\eqref{eq:Rdiff}.
\begin{lemma}
The choice of transmit vector in~\eqref{eq:params} results in the following effective channels at the legitimate receivers and the eavesdroppers:
\begin{equation}
\begin{aligned}
\rvy_j &= \bht_j^T \bb + \rvv_j \\
\rvz_k &= \bgt^T S_k \bb + \rvw_k
\end{aligned}\label{eq:RI-virtual}
\end{equation}
where the elements of $\bht_j \in \mathbb{R}^{LM}$ are rationally independent for each $j=1,\ldots, J_1$ and the elements of $\bgt$ belong to the set
\begin{equation}
\cA = \left\{\prod_{k=1}^{J_2}\prod_{i=1}^M g_{ki}^{\alpha_{ki}} \big|~  \al_{ki} \in \{0,\ldots, N\}\right\},\label{eq:RI-A}
\end{equation}
and the matrices $S_k \in \mathbb{R}^{(N+1)^{J_2  M} \times LM}$ have entries that are either $0$ or $1$ and each row has no more than $M$ ones. 
\label{lem:CWC:struct}
\end{lemma}
\begin{proof}
With the choice of $V$ in~\eqref{eq:RI-V}, from direct substitution,
\begin{equation}
\bht_j = \left[h_{j1}\bv^T,\ldots, h_{jM}\bv^T\right].
\end{equation}
Since the elements $h_{ji}$ are rationally independent, the elments of $\bv$ are rationally independent and independent of $h_{ji}$ it follows that all the elements of $\bht_j$ are rationally independent. Similarly we have that\begin{equation}\rvz_k = \bgt_k^T \bb +  \rvw_k,\end{equation}where \begin{equation}\bgt_k^T=\bg_k^T V = \left[g_{k1}\bv^T,\ldots,g_{kM}\bv^T\right].\label{eq:bgt}\end{equation} is a length $ML$ vector whose elements belong to the set $\cA$ in~\eqref{eq:RI-A}.

Since all elements of $\bgt_k$ belong to $\cA$ and $|\cA| = (N+1)^{J_2M}$ is, for sufficiently large $N$, smaller than $MN^{J_2M}$, it must that that the vector $\bgt_k$ has repeated elements.  Thus we can also express 
\begin{equation}
\bgt_k^T = \bg^T S_k,
\end{equation}
where $\bg$ is a vector consisting of all the elements in $\cA$ and $S_k \in \mathbb{R}^{(N+1)^{J_2 M} \times LM}$ is a matrix for which every column has exactly one element that equals $1$, and the remaining elements are zero. 

It remains to establish that each row in $S_k$ cannot have more than $M$ elements that equal $1$. Consider row $1$ in $S_k$. If the elements in columns $l_1,\ldots, l_T$ equal 1 then it follows that $g_{1} = \bgt_{k,l_1} = \ldots = \bgt_{k,l_T}$.  Thus to upper bound the number of ones in any given row, we  count the number of elements in the vector $\bgt_k$ in~\eqref{eq:bgt} that can be identical. Since each element in $\bv$ is distinct it follows that no two elements of the vector $g_{ki}\bv^T$ can be identical. Thus no more than $M$ elements in $\bgt_k$ can be identical, completing the claim. 
\end{proof}

The following lemma specifies the parameters of the PAM constellation~\eqref{eq:C-PAM} for the error probability at each legitimate receiver in the virtual channel~\eqref{eq:RI-virtual} to be arbitrarily small. The proof closely follows the proof of Prop.~\ref{prop:paramSelec} and is omitted. 
\begin{lemma}
Suppose that $\eps > 0 $ be an arbitrary constant and let $\g^2 = \frac{1}{M\sum_{t \in \cT}t^2}$. Select\begin{equation}\begin{aligned}Q &= P^{\frac{1-\eps}{2(ML+\eps)}}\\a &=\g\frac{P^{\frac{1}{2}}}{Q}\end{aligned}\label{eq:RI-QAParam}\end{equation} then we have that $||\bx||^2 \le P$ and for all channel vectors, except a set of measure zero, we have that
\begin{equation}
\pe \le \exp\left(-\eta P^\eps\right),\label{eq:pe}
\end{equation}
where $\eta$ is a constant that depends on the channel vector coefficients, but does not depend on $P$. 
\label{lem:pe}
\end{lemma}

To evaluate~\eqref{eq:Rdiff}, we will compute the terms $I(\bb;\rvy_j)$ and $I(\bb;\rvz_k)$ separately. From~\eqref{eq:pe} it follows via Fano's inequality that
\begin{align*}
H(\bb|\rvy_j) &= 1 + \pe H(\bb)\\
&= 1 + \pe LM \log(2Q+1)\\
&= 1 + o_P(1),
\end{align*}
where $o_P(1)$ denotes a function that goes to zero as $P \rightarrow \infty$. Thus we have that
\begin{align}
I(\bb;\rvy_j) &= H(\bb) - H(\bb|\rvy_j) \notag\\
&= LM \log(2Q+1) - 1 - o_P(1)\label{eq:indepSum}\\
&\ge LM \log Q - 1 - o_P(1) \notag\\
&= \frac{LM(1-\eps)}{2(LM+\eps)}\log P -1 -o_P(1)\label{eq:Qval}\\
&\ge \left(\frac{1}{2}-\eps\right)\log P   - 1 - o_P(1), \label{eq:mutInf1}
\end{align}where~\eqref{eq:indepSum} follows from the fact that each element of the vector $\bb$ is selected independently from $\cC$, and~\eqref{eq:Qval} follows by substituting the choice of $Q$ in~\eqref{eq:RI-QAParam} in Lemma~\ref{lem:pe}.

To upper bound the term $I(\bb;\rvz_k)$ we note that from~\eqref{eq:RI-virtual} it follows that $\rvz_k \rightarrow S_k \bb \rightarrow \bb$ holds. Hence
\begin{align}
I(\bb;\rvz_k) &\le I(\bb;S_k \bb) \notag\\
&= H(S_k \bb) \label{eq:RI-Ev1}\\
&\le \sum_{t=1}^{(N+1)^{MJ_2}} H(\{S_k\bb\}_t)\label{eq:RI-Ev2}\\
&\le (N+1)^{MJ_2}\log(2MQ+1)\label{eq:RI-Ev3}\\
&\le (N+1)^{MJ_2}\left\{\log Q + \log 4M \right\}\notag\\
&=(N+1)^{MJ_2}\left\{\frac{1-\eps}{2(ML + \eps)}\log P + \log 4M \right\}\label{eq:Qval2}
\end{align}
where~\eqref{eq:RI-Ev1} follows from the fact that $S_k\bb$ is a deterministic function of $\bb$, while~\eqref{eq:RI-Ev2} from the fact that conditioning reduces the entropy and finally~\eqref{eq:RI-Ev3} from the fact that each row of $S_k$ has at-most $M$ elements equal to 1 and the remaining elements equal to zero as stated in Lemma~\ref{lem:CWC:struct} and hence the support of $(S_k \bb)_t$ is $(-MQ,\ldots, MQ)$ and~\eqref{eq:Qval2} follows by substituting the value of $Q$ in~\eqref{eq:RI-QAParam}. Observe that
$$K_1 \defeq  (N+1)^{MJ_2}\log 4M$$ is a constant that does not depend on $P$ and substituting $L=N^{MJ_2}$,
\begin{equation}
\begin{aligned}
I(\bb;\rvz_k) &\le \frac{1}{2M(1+\eps)}\left(1 + \frac{1}{N}\right)^{MJ_2}\log P + K_1\\
&=\frac{1}{2M(1+\eps)}(1+o_N(1)) \log P + K_1,
\end{aligned}\label{eq:mutInf2}
\end{equation}
where $o_N(1)$ is a term that goes to zero as $N\rightarrow\infty$. Finally substituting~\eqref{eq:mutInf1} and~\eqref{eq:mutInf2} in~\eqref{eq:Rdiff} we have that
\begin{align}
\lim_{P\rightarrow\infty}\frac{R}{\frac{1}{2}\log P} &= 1 - 2\eps - \frac{1}{M(1+\eps)}(1+o_N(1))
\end{align}
which can be made arbitrarily close to $1-1/M$, by selecting $N$ sufficiently large and $\eps$ sufficiently close to zero.

\subsection{Proof of Prop.~\ref{prop:CWC:ML}}
\label{subsec:CWC:LB:PML}
The achievabilty scheme in this example employs a  multi-level coding scheme. We first propose a coding scheme for a linear deterministic channel over $\mathbb{F}_3$ and then extend it to the Gaussian case using multi-level coding. 

\subsubsection{Coding over a deterministic channel}
\begin{prop}
Consider a linear deterministic channel over $\mathbb{F}_3$ with two input symbols $\rvx_1$ and $\rvx_2$ and with output symbols described as follows:
\begin{equation}\begin{aligned}
&\rvy_1 = \rvx_1 + \rvx_2 \\
&\rvy_2 = \rvx_1 - \rvx_2\\
&\rvz_i = \rvx_i,\qquad i=1,2
\end{aligned}\label{eq:mLModel}\end{equation}where the addition and subtraction is defined over the group in $\mathbb{F}_3$. Then we can achieve a secrecy rate of $R = 1$~b/s for this channel. 
\label{prop:F3}
\end{prop}
\begin{proof}
The key idea behind the proof is to enable the legitimate receivers to take advantage of the field $\mathbb{F}_3$ in decoding while we limit the observation of the eavesdroppers to binary valued symbols. The wiretap code  is illustrated below:
\begin{equation}
\begin{array}{c|c}
 \mrm{msg.} & (\rvx_1,\rvx_2)    \\\hline
  0 & (0,0), (1,1)     \\
  1&   (0,1),(1,0)  \\\hline 
\end{array}\label{eq:code}\end{equation}
When message  bit $0$ needs to be transmitted the sender  selects one of the two tuples $(0,0)$ and $(1,1)$ at random and transmit the corresponding value of $(\rvx_1,\rvx_2)$. Likewise when bit $1$ needs to be transmitted one of the two tuples $(0,1)$ and $(1,0)$ will be transmitted. Note that when $b=0$ is transmitted $\rvy_1 \in \{0,2\}$ while $\rvy_2  = 0$ whereas when $b=1$ we have that $\rvy_1 =1$ and $\rvy_2 \in \{1,2\}$. It can be readily verified that each receiver is able to recover either message. Assuming that the messages are equally likely, it can also be readily verified that the message bit is independent of both $\rvx_1$ and $\rvx_2$ and thus the secrecy condition with respect to each eavesdropper is satisfied. 
 \end{proof}

\subsubsection{Multilevel Coding Scheme}
Fix integers $T$ and $M$ with the following properties: $T$ is the smallest integer such that for a given $\eps> 0$, $\Pr(\cE) \le \eps$ where $\cE$ denotes the error event\begin{equation}\cE =  \left\{\rvv_1, \rvv_2: \max_{i\in\{1,2\}}|\rvv_i|\ge 3^{T-1}\right\}\end{equation}and $M$ is the largest integer such that $3^{2M} \le P/2$. We  construct a multi-level code with a rate of  $M-T$ information bits and error probability at-most $\eps$.

Let the information bits be represented by the vector $\bb = (b_{T},\ldots, b_{M-1})$.  For each $i \in \{T,\ldots, M-1\}$,  we map the bit $b_i \in \{0,1\}$ into symbols  $(\xt_1(i),\xt_2(i))$ according to the code construction in~\eqref{eq:code}. The transmitted symbols are given by
\begin{equation}
\rvx_i = \sum_{l=T}^{M-1}\xt_i(l)3^l, \qquad i=1,2
\end{equation}
and the received symbols at the two receivers can be expressed as,
\begin{equation}
\rvy_1 = \sum_{l=T}^{M-1}\yt_1(l)3^l +\rvv_1,\qquad \rvy_2 = \sum_{l=T}^{M-1}\yt_2(l)3^l+\rvv_2
\end{equation}
where we have introduced
\begin{equation}
\begin{aligned}
\yt_1(l) &=\xt_1(l) + \xt_2(l) \in \{0,1,2\}, \\ \yt_2(l) &=\xt_1(l) - \xt_2(l) \in \{-1,0,1\}.
\end{aligned}\end{equation}
With the choice of $M$, it follows that $E[||\rvbx||^2] \le P$. Furthermore in the analysis of decoding, we declare an error if $\max_{i}|\rvv_i|> 3^{T-1}$. Conditioned on $\cE^c$,  note that
\begin{align}
&\rvy_i - \rvy_i~\mrm{mod} ~3^{T-1}\label{eq:comp}\\ &= \rvy_i  - \left(\sum_{l=T}^{M-1}\yt_i(l)~3^l\right)\mrm{mod}~3^T  - \rvv_1~\mrm{mod}~3^T \\
&= \rvy_i  - \rvv_1=\sum_{l=T}^{M-1}\yt_i(l)~3^l.\label{eq:remNoise}
\end{align} 
where we have used the fact that $\left(\sum_{l=T}^{M-1}\yt_i(l)~3^l\right)\mrm{mod}~3^T=0$ since each term in the summation is an integer multiple of $3^{T}$.

Thus by computing~\eqref{eq:comp} it is possible to retrieve $\sum_{l=T}^{M-1}\yt_i(l)~3^l$ (assuming the error event does not happen).  Since there is no carry over across levels we in turn retrieve $(\yt_i(T),\ldots, \yt_i(M))$ at each receiver.  Then applying the same decoding scheme{ as in Prop.~\ref{prop:F3} at each level, each receiver can recover the underlying bits $(b_T,\ldots, b_M)$.  If however we have that $|\rvv_i|\ge 3^{T-1}$,  then the above analysis leading to~\eqref{eq:remNoise} fails and an error is declared. Since $T$ is selected to be sufficiently large, this event happens with a probability that is less than $\eps$.

In order to complete the analysis it remains to show that 
$H(\bb|\rvz_i) = M-T$. We first enhance each eavesdropper by removing the noise variable  in~\eqref{eq:mLModel} i.e., $\zt_j = x_j$ for $j=1,2$.  Now consider
\begin{align*}
&H(b_T,\ldots, b_M|\zt_i) \\
=&H(b_T,\ldots, b_M|\xt_i(T),\ldots, \xt_i(M)) \\
=&\sum_{l=T}^M H(b_l|\xt_i(l))
=M-T,
\end{align*}
where the last relation follows from $H(b_l|\xt_i(l))=1$ since we use the code construction in~\eqref{eq:code} in mapping $b_l \rightarrow (\xt_1(l),\xt_2(l))$.

The resulting d.o.f.\ achieved by the multi-level code is given by
\begin{align}
d&=\frac{R}{\frac{1}{2}\log P}\\
&= \frac{M-T}{1/2\left(2M\log_2{3} +1\right)}\\
&= \log_3 2 + o_M(1),
\end{align}
where $o_M(1)\rightarrow 0$ as $M\rightarrow \infty$.

\section{Compound Wiretap Channel: Upper Bound}
\label{sec:CWC:UB}

In this section provide a proof of Theorem~\ref{thm:CWC:UB}. We use the convenient notation where the transmit vectors and received symbols are concatenated together i.e., $X = [\bx(1),\ldots, \bx(n)]$ and likewise $\rvby_j = [\rvy_j(1),\ldots,\rvy_j(n)]$, $\rvbz_k = [\rvz_k(1),\ldots,\rvz_k(n)]$ etc. In this notation the channel can be expressed as
\begin{equation}
\begin{aligned}
\rvby_j &= \bh_j^T X + \rvbv_j,\qquad j=1,\ldots,J_1\\
\rvbz_k &= \bg_k^T X + \rvbw_k,\qquad k=1,\ldots,J_2 \label{eq:model2}
\end{aligned}
\end{equation}

Note that it suffices to consider the case when $\min(J_1, J_2)\ge M$ in Theorem~\ref{thm:CWC:UB}. Otherwise the upper bound equals $1$, which continues to hold even in absence of secrecy constraints. 

Secondly we will assume that $J_1 = J_2 = M$ in deriving the upper bound. In all other cases, it is clear that the upper bound continues to hold as we only reduce the number of states.

For any code there exists a sequence $\eps_n$ that approaches zero as $n \rightarrow \infty$ such that
\begin{align}
&\frac{1}{n}I(\rvm;\rvbz_k) \le \eps_n, \qquad k=1,2,\ldots, M \label{eq:secrecy}\\
&\frac{1}{n}H(\rvm|\rvby_j) \le \eps_n,\qquad j=1,2,\ldots, M. \label{eq:fano}
\end{align}
where~\eqref{eq:secrecy} is a consequence of the secrecy constraint whereas~\eqref{eq:fano} is a consequence of Fano's inequality applied to receiver $j = 1,2,\ldots, J_r$. 

The proof is rather long and hence divided into the following subsections.

\subsubsection*{Upper bound from secrecy constraint}
\begin{lemma}
The rate of any compound wiretap code is upper bounded by the following expression:
\begin{equation}
nR \le \dent(\rvby_1,\ldots, \rvby_M) - \dent(\rvbz_k) + nc_k, \label{eq:UB1}
\end{equation}
where $c_k$ is a constant that does not depend on $P$. 
\label{lem:UB1}\end{lemma}
The proof of Lemma~\ref{lem:UB1} follows by considering the secrecy constriaint between each receiver and eavesdropper. A proof is provided in Appendix~\ref{app:UB1}.

Using~\eqref{eq:UB1} in Lemma~\ref{lem:UB1} for each eavesdropper $k=1,2,\ldots, M$ and adding up the resulting upper bounds we get that
\begin{align}
nR &\le \dent(\rvby_1,\ldots, \rvby_M) - \frac{1}{M} \sum_{k=1}^M \dent(\rvbz_k) + nc_0,\label{eq:UBk0}\\
&\le \dent(\rvby_1,\ldots, \rvby_M) - \frac{1}{M} \dent(\rvbz_1,\ldots, \rvbz_M) + nc_0,\label{eq:UBk1}\\
&\le \left(1-\frac{1}{M}\right)\dent(\rvby_1,\ldots,\rvby_M) + nd + nc_0,\label{eq:UBk2}
\end{align}
where $c_0 = \frac{1}{M}\sum_{k=1}^M c_k$ and $d$ are constants that do not depend on $P$. Here~\eqref{eq:UBk1} follows from the fact that conditioning reduced differential entropy and~\eqref{eq:UBk2} follows from 
\begin{equation}
\dent(\rvbz_1,\ldots,\rvbz_M) \ge \dent(\rvby_1,\ldots,\rvby_M) - nMd, \label{eq:LB}
\end{equation}
where $d$ is a constant that does not depend on $P$ as shown in Appendix~\ref{app:LB}.

\subsubsection*{Upper bound from multicast constraint}
We obtain the following upper bound on the joint entropy  $\dent(\rvby_1,\ldots,\rvby_M)$ 
\begin{lemma}
\begin{equation}
\dent(\rvby_1,\ldots,\rvby_M) \le \sum_{i=1}^M \dent(\rvby_i) - n(M-1)R + Mn\eps_n
\label{eq:JEUB}
\end{equation}
\label{lem:JEUB}
\end{lemma}
\begin{proof}
Our upper bound derivation uses the fact that the same message must be delivered to all the receivers and hence the output at the $M$ receivers must be sufficiently correlated. Note that
\begin{align}
&\dent(\rvby_1,\ldots,\rvby_M) =  \dent(\rvby_1,\ldots,\rvby_M|\rvm) + I(\rvm;\rvby_1,\ldots,\rvby_M)\notag\\
&\le \dent(\rvby_1,\ldots,\rvby_M|\rvm) + H(\rvm)\notag\\
&=\dent(\rvby_1,\ldots,\rvby_M|\rvm) + nR \label{eq:ma1}\\
&\le\sum_{j=1}^M\dent(\rvby_j|\rvm) + nR\label{eq:ma2}\end{align}\begin{align}
&=\sum_{j=1}^M \{\dent(\rvby_j)-H(\rvm)+H(\rvm|\rvby_j)\} + nR \notag\\
&\le\sum_{j=1}^M \dent(\rvby_j)- M H(\rvm)+ nM\eps_n + nR \notag\\
&\le \sum_{j=1}^M \dent(\rvby_j)- (M-1)nR + nM\eps_n  \label{eq:ma3}
\end{align}
where~\eqref{eq:ma1} and~\eqref{eq:ma3} follow from the fact that since the message is uniformly distributed over the set of size $2^{nR}$, it follows that $H(\rvm)= nR$ while~\eqref{eq:ma2} is a consequence of the fact that conditioning reduces differential entropy.
\end{proof}

\subsubsection*{Combining secrecy and multicast constraints}
In the final step we combine~\eqref{eq:UBk2} and~\eqref{eq:JEUB}. In particular, multipying both sides of~\eqref{eq:UBk2} by $M$ we get, with $d_0 = d+c_0$,
\begin{align}
&MnR = (M-1)\dent(\rvby_1,\ldots,\rvby_M) + nMd \notag\\
&\le (M-1)\left(\sum_{j=1}^M \dent(\rvby_j)- (M-1)nR + nM\eps_n \right) \notag\\ &\qquad + nMd_0 \\
&\le (M-1)\sum_{j=1}^M \dent(\rvby_j) \notag \\&\qquad - (M-1)^2nR + nM(d_0 + (M-1)\eps_n)\label{eq:UB-comb}
\end{align}
Rearranging the terms in~\eqref{eq:UB-comb}, we have that,
\begin{equation}
\begin{aligned}
&nR(M+(M-1)^2) \\&\qquad \le (M-1)\sum_{j=1}^M \dent(\rvby_j) + nM(d_0 + (M-1)\eps_n)\label{eq:UB-comb0}
\end{aligned}
\end{equation}
An upper bound for $\dent(\rvby_j)$ provided in Appendix~\ref{app:EntBnd}, (ref.~\eqref{eq:EntBnd} gives that
\begin{align}
&nR(M+(M-1)^2)\\ &\le (M-1)\left\{\sum_{j=1}^M
\frac{n}{2}\max(\log P,0) + nd_1\right\} \notag\\&\qquad+ nM(d_0 + (M-1)\eps_n)\\
&= (M-1)M
\frac{n}{2}\max(\log P,0)\notag\\&\qquad + nM(d_0 + (M-1)\eps_n + (M-1)d_1)
\end{align}
Thus with $d_2 = M(d_0 + (M-1)\eps_n + (M-1)d_1)$, a constant that does not depend on $P$ we have that
\begin{align}
R &\le \frac{M(M-1)}{M+(M-1)^2}\left(\frac{1}{2}\log P , 0\right) + d_2
\end{align}
which yields the desired upper bound on the degrees of freedom as stated in Theorem~\ref{thm:CWC:UB}.

\section{Compound Private Broadcast: Lower Bounds}
\label{sec:PBC:LB}
In this section we provide a proof for Theorem~\ref{thm:CPB:LB}.

When $\max(J_1,J_2) < M$ the transmitter achieves two degrees of freedom by zero-forcing the undesired groups. In particular, it finds two vectors $\bv_1$ and $\bv_2$ such that $\bh_j^T\bv_2 =0$ and $\bg_k^T \bv_1=0$ for $j=1,\ldots, J_1$ and $k=1,\ldots, J_2$. By transmitting $\bx = \bv_1 m_1 + \bv_2 m_2$, the effective channels at the two groups are given by\begin{equation}
\begin{aligned}
\rvy_j &= \bh_j^T \bv_1 m_1 + v_j \\
\rvz_k &= \bg_k^T \bv_2 m_2 + w_k \\
\end{aligned}
\end{equation}
Furthermore since, almost surely $\bh_j^T \bv_1 \neq 0$ and $\bg_k^T \bv_2 \neq 0$, it follows that one degree of freedom is achievable for each of the two groups.

To establish the degrees of freedom in the remaining two cases in Theorem~\ref{thm:CPB:LB}, we combine the real interference alignment scheme with wiretap coding. In particular we evaluate the following single-letter acheivable rate-pair for specific choice of auxiliary random variables that result from the real-interference alignment scheme. 
\begin{prop}
\label{prop:PB_achiev}
An achievable rate for the private memoryless broadcast channel $p_{\rvy_1,\ldots, \rvy_{J_1}, \rvz_1,\ldots, \rvz_{J_2} | \rvx}(\cdot)$ is as follows:
\begin{equation}
\begin{aligned}
R_1 &= \min_{j} I(\rvu_1; \rvy_j) - \max_{k}I(\rvu_1; \rvz_k, \rvu_2)\\
R_2 &= \min_{k} I(\rvu_2; \rvz_k) - \max_{j}I(\rvu_2; \rvy_j, \rvu_1),
\end{aligned}\label{eq:Rachiev}
\end{equation}where $(\rvu_1,\rvu_2)$ are \emph{mutually independent} random variables. The joint distribution satisfies the Markov condition
\begin{equation}
(\rvu_1,\rvu_2) \rightarrow \rvx \rightarrow (\rvy_1,\ldots, \rvy_{J_1},\rvz_1,\ldots, \rvz_{J_2}) 
\end{equation}
and $E[\rvx^2] \le P$.
\end{prop}

The proof of Prop.~\ref{prop:PB_achiev} is presented in Appendix~\ref{app:PB_achiev}.

\subsubsection*{Case $\min(J_1,J_2) < M \le \max(J_1, J_2)$}

We assume without loss of generality that $J_1 \ge M > J_2$. Let $\bv_1$ be a vector such that $\bg_k^T\bv_1 = 0$ for $k=1,\ldots, J_2$ and furthermore $\bh_j^T \bv_1 \neq 0$ for $j=1,\ldots, J_1$. The transmit vector $\bx$ is given by
\begin{equation}
\bx = \bv_1 m_1 + V_2 \bb_2 \label{eq:ds2x}
\end{equation}
where the precoding matrix $V_2$ and information symbols $m_1$ and $\bb_2$ are selected based on the real interference alignement scheme as described below. Let $N_2 \in \mathbb{N}$ be a sufficiently large integer and define
\begin{align}
\cT_2 &= \left\{\prod_{i=1}^M\prod_{j=1}^{J_1} h_{ji}^{\beta_{ji}}
~\bigg|~0 \le \beta_{ji} \le N_2-1\right\}\label{eq:dsT2}\\
\cA_2 &= \left\{\prod_{i=1}^M\prod_{j=1}^{J_1} h_{ji}^{\beta_{ji}}
~\bigg|~0 \le \beta_{ji} \le N_2\right\}\label{eq:dsA2}
\end{align}consisting of $L_2 = N_2^{MJ_1}$ and $L_2' = (N_2 + 1)^{MJ_1}$ elements respectively. Let $\bv_2 \in \mathbb{R}^{L_2 \times 1}$ be a vector consisting of all  elements in $\cT_2$ and let
\begin{equation}
V_2 = \left[\begin{array}{cccc} \bv_2^T & 0 &\ldots & 0 \\ 0 & \bv_2^T & \ldots & 0 \\ \vdots & \vdots &\ddots & \vdots \\ 0 & 0 &\ldots & \bv_2^T \end{array}\right] \in \mathbb{R}^{M \times ML_2}.
\end{equation}
The elements $b_{21},\ldots, b_{2\{ML_2\}}$ of the information vector $\bb_2 \in \mathbb{R}^{ML_2}$ for group 2 in~\eqref{eq:ds2x} are sampled independently and uniformly from the PAM constellation of the form \begin{equation}
\cC = a\left\{-Q,-Q+1,\ldots, Q-1, Q\right\}\label{eq:dsPAM}
\end{equation} while the information symbol for group $1$ is of the form 
\begin{equation}
\rvm_1 = \bal^T \bb_1 = [\begin{array}{ccc}\al_1 &\ldots & \al_{(M-1)L_2}\end{array}]\left[\begin{array}{c}b_{11}\\ \vdots \\ b_{1{(M-1)L_2}}\end{array}\right], 
\end{equation}where the elements $\al_j$ are selected to be rationally independent of all other coefficients and their monomials and the symbols in $\bb_1$ are also sampled independently and uniformly from the PAM constellation~\eqref{eq:dsPAM}. 

Substituting the choice of $\bx$ in~\eqref{eq:ds2x} the channel model reduces to\begin{equation}\begin{aligned}
\rvy_j  &= (\bh_j^T\bv_1) \bal^T \bb_1 + \bh_j^T V_2 \bb_2 + \rvv_j\\
\rvz_k &= \bg_k^T V_2 \bb_2 + \rvw_k
\end{aligned}\label{eq:ds2Vc}\end{equation}
Following the same line of reasoning as in the proof of Lemma~\ref{lem:CWC:struct}, our choice of $\rvbx$ in~\eqref{eq:ds2x} reduces the channel model as stated below.
\begin{lemma}
With choice of $\bx$ in~\eqref{eq:ds2x} the output symbols at the receivers can be expressed as
\begin{equation}\begin{aligned}
\rvy_j  &= \bht_j^T\bb_1 + \bht^T T_j \bb_2 + \rvv_j\\
\rvz_k &= \bgt_k^T \bb_2 + \rvw_k
\end{aligned}\label{eq:ds2Vc2}\end{equation}
where $\bht \in \mathbb{R}^{L_2'}$ is a vector consisting of all the elements in $\cA_2$, (cf~\eqref{eq:dsA2}), and the entries of vectors $\bht_j$  and $\bgt_k$ are rationally independent and also independent of  elements in $\cA_2$. The entries of matrix $T_j \in \mathbb{R}^{L_2' \times ML_2}$ are either 0 or 1 and there are no more than $M$ ones in each row of $T_j$.
\label{lem:dsVc2}
\end{lemma}
The following Lemma can be established along the lines of Prop.~\ref{prop:paramSelec}.
\begin{lemma} Suppose that $\eps > 0 $ be an arbitrary constant and let $\g^2 = \frac{1}{\sum_{t \in \cT_2}t^2 + \sum_{j=1}^{(M-1)L_2}\al_j^2}$ be a normalizing constant that does not depend on $P$. If we select\begin{equation}\begin{aligned}Q &= \left(\frac{P}{2}\right)^{\frac{1-\eps}{2((M-1)L_2+L_2'+\eps)}}\\a &=\g\frac{\left(\frac{P}{2M}\right)^{\frac{1}{2}}}{Q}\end{aligned}\label{eq:RIds2-QAParam1}\end{equation} then we have that $||\bx||^2 \le P$ and for all channel vectors, except a set of measure zero, we have that
\begin{equation}
\pe \le \exp\left(-\eta P^\eps\right),\label{eq:ds2pe1}
\end{equation}
where $\eta$ is a constant that depends on the channel vector coefficients, but does not depend on $P$. 
\label{lem:ds2pe1}
\end{lemma}

The achievable rate pair $(R_1,R_2)$ is obtained by evaluating~\eqref{eq:Rachiev} in Prop.~\ref{prop:PB_achiev}  
with $\rvu_1 = \bb_1$, $\rvu_2 = \bb_2$ and with $\rvbx$ in~\eqref{eq:ds2x}.
\begin{align}
&R_1 = \min_j I(\bb_1; \rvy_j) - \max_k I(\bb_1;\rvz_k,\bb_2) \notag\\
&= \min_j I(\bb_1; \rvy_j)\label{eq:ds2i1}\\
&= H(\bb_1) - \max_j H(\bb_1|\rvy_j) \notag\\
&= H(\bb_1) - o_P(1)\label{eq:ds2i2}\\
&= \sum_{i=1}^{(M-1)L_2}H(b_{1i})-o_P(1) \label{eq:ds2i3}\\
&= (M-1)L_2\log(2Q+1) -o_P(1)\label{eq:ds2i4a}\\
&\ge (M-1)L_2\log(Q) -o_P(1)\notag\\
&= (M-1)L_2\frac{1-\eps}{2((M-1)L_2 + L_2'+ \eps)}\log P - o_P(1) \label{eq:ds2i5}
\end{align} where~\eqref{eq:ds2i1} follows from the fact that $(\bb_2, \rvz_k)$ are independent of $\bb_1$ in~\eqref{eq:ds2Vc2},
and~\eqref{eq:ds2i2} follows from Fano's inequality via~\eqref{eq:ds2pe1}, and~\eqref{eq:ds2i3} and~\eqref{eq:ds2i4a} follow from the fact that the entries of $\bb_1$ are selected independently and uniformly from the constellation $\cC$ in~\eqref{eq:dsPAM} and finally~\eqref{eq:ds2i5} follows by substituting the expression for $Q$ in~\eqref{eq:RIds2-QAParam1}. 

\begin{align}
&R_2 = \min_k I(\bb_2; \rvz_k) - \max_j I(\bb_2;\rvy_j,\bb_1) \notag\\
&\ge \min_k I(\bb_2; \rvz_k) - \max_j I(\bb_2; T_j \bb_2) \label{eq:ds2j1}\\
&\ge H(\bb_2) - \max_{k} H(\bb_2|\rvz_k) - \max_j H(T_j\bb_2) \label{eq:ds2j2}\\
&\ge H(\bb_2) - o_P(1) - \max_j H(T_j\bb_2) \label{eq:ds2j2}\\
&= ML_2 \log(2Q+1) - o_P(1) - \max_j H(T_j\bb_2) \label{eq:ds2j3}\\
&\ge ML_2 \log(2Q+1) -  o_P(1) - \max_j \sum_{k=1}^{L_2'}H(\{T_j\bb_2\}_k) \label{eq:ds2j4}\end{align}
where~\eqref{eq:ds2j1} follows from the fact that since $\bb_2$ and $\bb_1$ are indpendent, we have the Markov chain
$(\bb_2,\rvz_k) \rightarrow T_j \bb_1 \rightarrow \bb_1$ for $\rvz_k$ in \eqref{eq:ds2Vc2}, and~\eqref{eq:ds2j2} follows from Fano's inequality via~\eqref{eq:ds2pe1}, and
~\eqref{eq:ds2j3} follows from the fact that the entries of $\bb_2$ are i.i.d.\ and uniformly distributed over $\cC$ in~\eqref{eq:dsPAM}, and~\eqref{eq:ds2j4} follows from the fact that conditioning reduces entropy.

We further simplify the last term in~\eqref{eq:ds2j4} as follows:
\begin{align}
H(\{T_j\bb_2\}_k) &\ge L_2'\log(2MQ+1) \label{eq:ds2j5}\\
&\ge L_2'\log(2Q+1) + L_2'\log M\label{eq:ds2j4a}
\end{align}
~\eqref{eq:ds2j5} from the fact that as stated in Lemma~\ref{lem:dsVc2} each row of the matrix $T_j$ has no more than $M$ ones and thus the support of each element $\{T_j\bb_2\}_k$ is contained in $\{-MQ,\ldots, MQ\}$. Substiuting~\eqref{eq:ds2j4a} into~\eqref{eq:ds2j4} and defining $K = o_P(1) + L_2'\log M$, a constant that does not depend on $P$, we have that
\begin{align}
R_2&\ge (ML_2 - L_2') \log(2Q+1) - K \\
&\ge (ML_2 - L_2') \log Q - K \notag\\
&\ge (ML_2 - L_2') \frac{1-\eps}{2((M-1)L_2 + L_2' + \eps)}\log P - K \label{eq:ds2j6}\end{align}
where the last term follows by substituting the expression for $Q$ in~\eqref{eq:RIds2-QAParam1}.
Using~\eqref{eq:ds2i5} and~\eqref{eq:ds2j6} we have that
\begin{align}
&\lim_{P\rightarrow\infty}\frac{R_1+R_2}{\frac{1}{2}\log P} \notag\\
&= (1-\eps)\frac{(M-1)L_2 + ML_2 - L_2'}{(M-1)L_2 + L_2' + \eps}\\
&= (1-\eps)\frac{2M - 1 - \left(1+\frac{1}{N_2}\right)^{MJ_1}}{(M-1)+ \left(1 + \frac{1}{N_2}\right)^{MJ_1} + \frac{\eps}{N_2^{MJ_1}}}\label{eq:ds2t2}
\end{align}
where we have substituted $L_2 = N_2^{MJ_1}$ and $L_2' = (N_2+1)^{MJ_1}$ in the last expression.  Finally note that the expression~\eqref{eq:ds2t2} can be made arbitrarily close to $2\frac{M-1}{M}$ by choosig $N_2$ sufficiently large and $\eps$ sufficiently close to zero. 

\subsubsection*{Case $\min(J_1, J_2) \ge M$}

When $\min(J_1,J_2)\ge M$, we need to do signal alignment to both groups of users. Let $J = \max(J_1,J_2)$. We design our scheme assuming $J_1 = J_2 = J$. Clearly this coding scheme can also be used in the original case.  We define
\begin{align}
\cT_1 &= \left\{\prod_{i=1}^M\prod_{j=1}^{J} g_{ji}^{\al_{ji}}
~\bigg|~1 \le \al_{ji} \le N \right\}\label{eq:dsT1}\\
\cT_2 &= \left\{\prod_{i=1}^M\prod_{j=1}^{J} h_{ji}^{\al{ji}}
~\bigg|~1 \le \al_{ji} \le N \right\}\label{eq:dsT2a}\end{align}\begin{align}
\cA_1 &= \left\{\prod_{i=1}^M\prod_{j=1}^{J} g_{ji}^{\al_{ji}}
~\bigg|~1 \le \al_{ji} \le N+1\right\}\label{eq:dsA1}\\
\cA_2 &= \left\{\prod_{i=1}^M\prod_{j=1}^{J} h_{ji}^{\al_{ji}}
~\bigg|~1 \le \al_{ji} \le N+1\right\}\label{eq:dsA1}
\end{align}where the sets $\cT_1$ and $\cT_2$ consist of $L = N^{MJ}$ elements whereas the sets $\cA_1$ and $\cA_2$ consist of $L' = (N + 1)^{MJ}$ elements . Let $\bv_1,\bv_2 \in \mathbb{R}^{L \times 1}$ be vectors consisting of all  elements in $\cT_1$ and $\cT_2$ respectively and let
\begin{equation}
V_k = \left[\begin{array}{cccc} \bv_k^T & 0 &\ldots & 0 \\ 0 & \bv_k^T & \ldots & 0 \\ \vdots & \vdots &\ddots & \vdots \\ 0 & 0 &\ldots & \bv_k^T \end{array}\right],\qquad k=1,2
\end{equation}
be the precoding matrices and let the transmit vector be expressed as
\begin{equation}
\bx = V_1 \bb_1 + V_2 \bb_2,\label{eq:ds3bx}
\end{equation}where the vectors $\bb_1, \bb_2 \in \mathbb{R}^{ML}$ consist of information symbols for group $1$ and $2$ respectively. Each entry in these vectors is sampled independently and uniformly from the PAM constellation of the form
\begin{equation}
\cC = a\left\{-Q,\ldots, Q\right\}.\label{eq:ds3PAM}
\end{equation}
Following the line of reasoning in Lemma~\ref{lem:CWC:struct} and Lemma~\ref{lem:dsVc2} we have the following.
\begin{lemma}
With the choice of $\bx$ in~\eqref{eq:ds3bx} we can express the resulting channel output symbols at each receiver as follows,
\begin{equation}
\begin{aligned}
\rvy_j &= \bht_j^T \bb_1 + \bht^T T_j \bb_2 + \rvv_j\\
\rvz_k &= \bgt_k^T\bb_2 + \bgt^T S_k \bb_2 + \rvw_k,
\end{aligned}\label{eq:ds3vc}
\end{equation}where the vectors $\bht, \bgt \in \mathbb{R}^{L'}$ consist of all elements belonging to the sets $\cA_2$ and $\cA_1$ respectively. The elements of vector $\bht_j$ are rationally independent and independent of the the elements of $\cA_1$ and likewise the elements of $\bgt_k$ are rationally independent and independent of the elements of $\cA_2$. The matrices $T_j, S_k \in \mathbb{R}^{L' \times M L}$ have elements that take values of either $0$ or $1$ and there are no more than $M$ elements whose value equals $1$ in any given row of these matrices.
\label{lem:chanReduction}\end{lemma}
The choice of parameters $a$ and $Q$ stated below can be derived along the lines of Prop.~\ref{prop:paramSelec}.
\begin{lemma}
Suppose that $\eps > 0 $ be an arbitrary constant and let $\g^2 = \frac{1}{\sum_{t \in \cT_2}t^2 + \sum_{t\in \cT_1}t^2 }$ be a normalizing constant that does not depend on $P$. If we select\begin{equation}\begin{aligned}Q &= \left(\frac{P}{2}\right)^{\frac{1-\eps}{2(ML+L'+\eps)}}\\a &=\g\frac{\left(\frac{P}{2M}\right)^{\frac{1}{2}}}{Q}\end{aligned}\label{eq:RIds3-QAParam1}\end{equation} then we have that $||\bx||^2 \le P$ and for all channel vectors, except a set of measure zero, we have that
\begin{equation}
\pe \le \exp\left(-\eta P^\eps\right),\label{eq:ds2pe1}
\end{equation}
where $\eta$ is a constant that depends on the channel vector coefficients, but does not depend on $P$. 
\label{lem:ds3pe}
\end{lemma}
Finally to compute the achievable rate pair we substitute in~\eqref{eq:Rachiev}, $\rvu_1 = \bb_1$, $\rvu_2 = \bb_2$ and $\rvx$ as specified in~\eqref{eq:ds3bx}. Following analogous calculations that lead to~\eqref{eq:ds2j6} we have that
\begin{align}
R_1 &\ge \frac{(ML-L')(1-\eps)}{ML + L' + \eps}-K\\
R_1 &\ge \frac{(ML-L')(1-\eps)}{ML + L' + \eps}-K,
\end{align}
where $K$ is a constant that does not depend on $P$. Substituting $L = N^{MJ}$ and $L' = (N+1)^{MJ}$,\begin{align}
\lim_{P\rightarrow \infty}\frac{R_1+R_2}{\frac{1}{2}\log P} = 2(1-\eps)\frac{M - \left(1+\frac{1}{N}\right)^{MJ}}{M + \left(1+\frac{1}{N}\right)^{MJ} + \frac{\eps}{N^{JM}}}\end{align}which can be made arbitrarily close to $2\frac{M-1}{M+1}$ by selecting $N$ sufficiently large and $\eps$ sufficiently close to zero.

\section{Compound Private Broadcast: Upper Bounds}
\label{sec:PBC:UB}
When $\max(J_1, J_2) \le M$, the stated upper bound is $2$. It holds even when $J_1=J_2=1$.

When $\min(J_1,J_2)< M \le \max(J_1, J_2)$ we assume without loss of generality that $J_1 < J_2$. The upper bound is developed assuming $J_1 =1$ and $J_2 = M$.
Any sequence of codes that achieves a rate-pair $(R_1, R_2)$ satisfies, via Fano's inequality,
\begin{align}
\frac{1}{n}H(\rvm_1|\rvby_1) \le \eps_n, \frac{1}{n}H(\rvm_2|\rvbz_k) \le \eps_n, k=1,\ldots, M,
\label{eq:sumFano}\end{align}
and the secrecy constraints
\begin{equation}
\frac{1}{n}I(\rvm_2; \rvby_1)\le \eps_n, \frac{1}{n}I(\rvm_1;\rvbz_k) \le \eps_n,  k=1,\ldots, M,
\end{equation}
We can upper bound the sum-rate of the messages as:
\begin{align}
&n(R_1+R_2) \le I(\rvm_1, \rvm_2;\rvby_1, \rvbz_1,\ldots, \rvbz_M) + n\eps_n\notag\\
&= \dent(\rvbz_1,\ldots, \rvbz_M) + \dent(\rvby_1|\rvbz_1,\ldots, \rvbz_M) \notag\\
&\qquad- \dent(\rvby_1, \rvbz_1,\ldots, \rvbz_M|\rvm_1,\rvm_2)+n\eps_n\label{eq:diff}
\end{align}
\iffalse
&= h(\rvbz_1,\ldots, \rvbz_M) + h(\rvby_1|\rvbz_1,\ldots, \rvbz_M) \notag\\
&\qquad- h(\rvby_1, \rvbz_1,\ldots, \rvbz_M|\rvm_1,\rvm_2,\rvbx)+n\eps_n\notag\\
&= h(\rvbz_1,\ldots, \rvbz_M) + h(\rvby_1|\rvbz_1,\ldots, \rvbz_M) \notag\\
&\qquad- h(\rvbv_1, \rvbw_1,\ldots, \rvbw_M)+n\eps_n\notag\\
\fi
\
Since the channel vectors $\bg_1,\ldots, \bg_M$ are linearly independent, we can express
$$\bh_1 = \sum_{k=1}^M \lambda_k \bg_k$$
and hence
\begin{align}
&\dent(\rvby_1|\rvbz_1,\ldots, \rvbz_M) \le \dent\left(\rvby_1 - \sum_{i=1}^M \lambda_i \rvbz_i\right) \\
&=\dent\left(\rvbv_1 - \sum_{i=1}^M \lambda_i \rvbw_i\right) = nK_1
\end{align}
where $K_1$ is a constant that does not depend on $P$. We also have that
\begin{align}
&\dent(\rvby_1, \rvbz_1,\ldots, \rvbz_M|\rvm_1,\rvm_2) \\&\ge \dent(\rvby_1, \rvbz_1,\ldots, \rvbz_M|\rvm_1,\rvm_2,X) \\
&=\dent(\rvbv_1,\rvbw_1,\ldots, \rvbw_M) = nK_2,
\end{align}
where $K_2$ is a constant that does not depend on $P$. From~\eqref{eq:diff} with $K_3 = K_1 -K_2 +\eps_n$, we have that
\begin{align}
n(R_1 + R_2) &\le \dent(\rvbz_1,\ldots, \rvbz_M) + nK_3\label{eq:sumUB1}.
\end{align}

Since message $2$ needs to be delivered to $M$ receivers, we have from  Lemma~\ref{lem:JEUB}, that
\begin{equation}
\label{eq:sumUB2}
(M-1)nR_2 \le \sum_{i=1}^M \dent(\rvbz_i) - \dent(\rvbz_1,\ldots, \rvbz_M) +n \eps_n,
\end{equation}
and it follows from~\eqref{eq:sumFano} that
\begin{align}
nR_1 &\le I(\rvm_2;\rvby_1) \notag\\
&= \dent(\rvby_1) -\dent(\rvby_1|\rvm_2)\notag\\
&=\dent(\rvby_1)  + K_4 \label{eq:sumUB3}
\end{align}
Combining~\eqref{eq:sumUB1},~\eqref{eq:sumUB2} and~\eqref{eq:sumUB3} we have that
\begin{align}
nM(R_1+R_2) &\le \sum_{i=1}^M\dent(\rvbz_i) + (M-1)\dent(\rvby_1) + nK_5,
\end{align}
where $K_5 = (M-1)K_4 + K_3$ is a constant that does not depend on $P$. 
Using the upper bound in Appendix~\ref{app:EntBnd} on the entropy of received vector we have that
\begin{equation}
nM(R_1+R_2) \le (2M-1)\frac{n}{2}\max\left(\log P, 0\right) + nK_6,
\end{equation}
and hence,
\begin{equation}
\lim_{P\rightarrow\infty}\frac{R_1 + R_2}{\frac{1}{2}\log P} = \frac{2M-1}{M},
\end{equation}as required.

In the final case when $\min(J_1,J_2)\ge M$, we develop the upper bound assuming that $J_1 = J_2 = M$. The upper bound continues to hold when $J_1 \ge M$ and $J_2 \ge M$ as we are only reducing the number of states. 
For any private broadcast code, there exists a sequence $\eps_n$ such that 
\begin{align}
\frac{1}{n}H(\rvm_1|\rvbz_k) \le \eps_n,\quad \frac{1}{n}I(\rvm_1;\rvbz_k) \le \eps_n
\end{align}
It follows via~\eqref{eq:UB1} in Lemma~\ref{lem:UB1} that
\begin{equation}
nR_1 \le \dent(\rvby_1,\ldots, \rvby_M) - \dent(\rvbz_k) + nc_k, \label{eq:UBp1}
\end{equation}
where $c_k$ is a constant that does not depend on $P$. 
Similarly applying Fano's Inequality
\begin{equation}
\frac{1}{n}H(\rvm_2|\rvbz_k) \le \eps_n
\end{equation}
for message $\rvm_2$ we can establish along the lines of Lemma~\ref{lem:JEUB}, that
\begin{equation}
\dent(\rvbz_1,\ldots,\rvbz_M) \le \sum_{i=1}^M \dent(\rvbz_i) - n(M-1)R_2 + Mn\eps_n
\label{eq:JEUBp}
\end{equation}
Furthermore, using~\eqref{eq:LB} we have that for a constant $d$ that does not depend on $P$, we have that
\begin{align}
&\dent(\rvby_1,\ldots,\rvby_M) \le \dent(\rvbz_1,\ldots,\rvbz_M) + nMd \\
&\le \sum_{i=1}^M \dent(\rvbz_i) - n(M-1)R_2 + Mn\eps_n + nMd, \label{eq:UBp2}
\end{align}
where the last relation follows by substituting~\eqref{eq:JEUBp}. Combining~\eqref{eq:UBp1} and~\eqref{eq:UBp2} and rearranging, we have that
\begin{align}
&nR_1 + n(M-1)R_2 \\&\le \sum_{i=2}^M\dent(\rvbz_i) + nMd + nc_1 + Mn\eps_n\\
&\le (M-1)\frac{n}{2}\max\left(\log P, 0\right) + nK, \label{eq:term1}
\end{align}
where $nK = n(M-1) d_1 + nMd + Mn\eps_n$ is a constant that does not depend on $P$. Note that the last relation follows by the upper bound on the entropy of each received vector as derived in Appendix~\ref{app:EntBnd}. Using a symmetric argument it follows that
\begin{align}
nR_2 + n(M-1)R_1 &\le (M-1)\frac{n}{2}\max\left(\log P, 0\right) + nK, \label{eq:term2}.
\end{align}
Combining~\eqref{eq:term1} and~\eqref{eq:term2}, we have that
\begin{equation}
\lim_{P\rightarrow\infty}\frac{R_1+R_2}{\frac{1}{2}\log P} \le 2\frac{M-1}{M}
\end{equation}
as required.

\section{Conclusions}
\label{sec:concl}
This paper develops new upper and lower bounds on the degrees of freedom of the compound wiretap channel. The upper bound is developed through a new technique that captures the tension between the secrecy and common message constraints and strictly improves the pairwise upper bound. A lower bound, that achieves non-vanishing degrees of freedom for arbitrary number of receiver states, is established based on the real interference alignment technique.  These techniques  are extended to a related problem: the private broadcast channel and again new upper and lower bounds on the degrees of freedom are established.

Our results suggest that interference alignment can potentially play a significant role in designing robust physical layer secrecy protocols. This  technique provides the mechanism to reduce the number of  dimensions occupied by an interfering signal at a receiver, thus increasing the number dimensions that are available for the signal of interest. We apply this technique  to  reduce the observed signal dimensions at multiple eavesdroppers, thus significantly enhancing the rates achieved for the compound wiretap channel compared to traditional techniques.  Nevertheless, we illustrate by an example that unlike the $K-$ user interference channel~\cite{realInterf1} and the compound MIMO broadcast channel~\cite{compoundMIMO1, compoundMIMO2}, a direct application of interference alignment cannot achieve the secrecy capacity of the compound wiretap channel. In terms of future work it will be interesting to close the gap between the upper and lower bounds. Another promising direction is to investigate recent ideas on  practical techniques based on reconfigurable antennas~\cite{jafar:10} for the compound  wiretap channel model.

\appendices

\section{Proof of Prop.~\ref{prop:CWC:TSLB}}
\label{app:CWC:TSLB}

It suffices to consider the case when $\min(J_1,J_2)\ge M$, since the other case is covered in Theorem~\ref{thm:CWC:LB}.
We separately show how to attain $\frac{M-1}{J_1}$ and $\frac{M-1}{J_2}$ degrees of freedom. 

\subsection{Attaining $\frac{M-1}{J_1}$ degrees of freedom}
Our coding scheme is described as follows.
\begin{enumerate}
\item Let $T= {{J_1} \choose{M-1}}$ denote all possible subsets of users of size $M-1$.  We label these subsets as $\cS_1,\ldots,\cS_T$. Note that each user belongs to $T_0 ={{J_1-1}\choose{M-2}} $ subsets. 

\item Let $n$ be a sufficiently large integer. A message $\rvm$ consists of $n T_0 R_0$ information bits where\begin{equation}R_0 = \frac{1}{2}\log P - \Theta, \label{eq:Ts1}\end{equation} and where $\Theta$ is a sufficiently large constant, (which will be specified later) that does not grow with $P$. The message is mapped into a codeword of a $(T,T_0)$ erasure code $\cC$ i.e., $m \rightarrow (m_1, m_2,\ldots, m_{T})$ where each symbol $m_i$ consists of $nR_0$ information bits. Each receiver  retrieves the message $m$ provided it observes any $T_0$ symbols. Furthermore as established in~\cite{khistiTchamWornell:07} suitable wiretap code constructions exist such that provided each of the $T$ symbols are individually protected, the overall message remains protected. i.e.,
\begin{multline}
I(\rvm_t;\rvz_k^n) \le n\eps_n, \qquad \forall t = 1,\ldots, T  \\ \Rightarrow I(\rvm;\rvz_k^n) \le Tn\eps_n
\label{eq:secMC}\end{multline}

The overall rate  $R = \frac{T_0}{T} R_0 $ results in the following degrees of freedom:
\begin{align*}
d &= \frac{T_0}{T}=\frac{{J_1 - 1 \choose M -2}}{{J_1 \choose M-1}}\\
&= \frac{M-1}{J_1}
\end{align*}  as required.

It remains to show how to transmit message $\rvm_t$ such that each user in a subset $\cS_t$  decodes it with high probability while satisfying $I(\rvm_t;\rvz_k^n) \le n\eps_n$.

\item  Each subset $\cS_t$ is served over $n$ channel uses. The message $\rvm_t$ is transmitted to $M-1$ users belonging to this subset along the lines of Theorem~\ref{thm:CWC:LB} when $J_1 < M$ i.e.,  by transmitting information symbols in the common range space and noise in the common null space of these users  (c.f.~\eqref{eq:LowUsers}). 
\begin{equation}
\bx_t = \bu_t \rvs + \ba_t{ \rvn},
\end{equation}
where $\ba_t$ and $\bu_t$ are unit norm vectors such that  $\bh_i^T \bu_t \neq 0$  and $\bh_i^T \ba_t = 0$ for each $i \in \cS_t $ and $\rvs$ and $\rvn$ are information bearing and noise symbols respectively.

Following the analysis leading to~\eqref{eq:TsRate} we can see that the following rate is achievable:
\begin{equation}
\begin{aligned}
R_t &= \frac{1}{2}\log P - \Theta_t, \\ \Theta_t  &= - \min_{i \in \cS_t }\frac{1}{2}\log |\bh_i^T\bu_t|^2 + \\&\qquad \max_{k  }\frac{1}{2}\log\left(1 + \frac{|\bg_k^T\bu_t|^2}{ |\ba_t^T\bg_k|^2}\right) ,
\end{aligned}\end{equation}
where $\Theta_t$ is a constant that does not scale with $P$. Furthermore we let $\Theta$ in~\eqref{eq:Ts1} to be
\begin{equation}
\Theta = \max_{t} \Theta_t.
\end{equation}  
\item With the choice of rate in~\eqref{eq:Ts1} every user in each subset $\cS_t$ can decode the message $\rvm_t$ with high probability. Each user will have access to $T_0$ elements of the codeword $(\rvm_1,\ldots, \rvm_T)$ and hence   recover the original message $\rvm$.  Furthermore each individual message is protected from each eavesdropper\begin{equation}
I(\rvm_t; \rvz_k^n) \le n\eps_n
\end{equation} and hence from~\eqref{eq:secMC} it follows that $I(m;\rvz_k^n)\le nT\eps_n$. Since $\eps_n$ can be made sufficiently small, the secrecy condition is satisfied.  
\end{enumerate}

\subsection{Attaining $\frac{M-1}{J_2}$ degrees of freedom} 
Our coding scheme is described as follows.
\begin{enumerate}
\item We considers all possible $T^e= {{J_2}\choose{M-1}}$ subsets of $M-1$ eavesdroppers and label them as $\cS^2_1,\ldots, \cS^2_{T^e} $. Note that each eavesdropper belongs to a total of $T^e_1 = {{J_2-1}\choose{M-2}}$ subsets.
\item Consider a parallel noise-less wiretap channel consisting of $T^e$ links,  where each link supports a rate $nR_1$, where \begin{equation}R_1 = \frac{1}{2}\log P - \Omega\end{equation} and $\Omega$ is a sufficiently large constant that will be specified later.  Each eavesdropper is absent on a total of $T^e_1$ links while each legitimate receiver observes all the $T^e$ links. Following the scheme in~\cite{khistiTchamWornell:07} we can transmit a message $\rvm$ of rate $nR_1{T^e_1}$ by mapping the message $\rvm\rightarrow(\rvm_1,\ldots, \rvm_{T^e})$.  The symbol  $\rvm_k$, consists of $nR_1$ bits and  forms the input message on channel $k$.  

\item  For each choice of $\cS^2_t$, we transmit information in the common null-space of the eavesdroppers in this selected set. Let $\bb_t$ be a vector such that $\bb_t^T \bg_j = 0$ for each $j \in \cS^2_t$ and transmit $$\bx_t = \bb_t \rvs,$$ where $\rvs$ is the information bearing symbol. Since each vector $\bh_i$ is linearly independent of any collection of $M-1$ eavesdropper channel vectors it follows that $\bh_i^T\bb_t \neq 0$, and one can achieve a rate
\begin{multline}R_1 = \frac{1}{2}\log\left(1 + \min_{\substack{t \in \{1,\ldots, T^e\}\\ i \in \{1,\ldots, J_1\}}} |\bh_i^T \bb_t|^2 P\right)\\ \ge \frac{1}{2}\log P - \Omega,\end{multline}where\begin{equation}\Omega = -\min_{\substack{t \in \{1,\ldots, T^e\}\\ i \in \{1,\ldots, J_1\}}} \frac{1}{2}\log|\bh_i^T \bb_t|^2. \end{equation} 

With this choice of $\Omega$, each receiver decodes each of the messages $m_1,\ldots, m_T$ with high probability. Furthermore, each of the eavesdropper does not have access to $T_1$ sub-messages corresponding to the subsets $\cS_t$ to which it belongs. By virtue of our code construction, this ensures that $I(\rvm;\rvz_k^n)\le n\eps_n$.

The overall achievable rate is given by $R = \frac{T^e_1}{T^e}R_1$ and hence the achievable degrees of freedom  are given by\begin{align*}d &=\frac{T_1^e}{T^e}= \frac{{J_2-1\choose M-2}}{{J_2 \choose M-1}}= \frac{M-1}{J_2}\end{align*}
as required.\end{enumerate}\

\section{Proof of Lemma~\ref{lem:UB1}}
\label{app:UB1}
From the secrecy constraint and Fano's inequality (c.f.~\eqref{eq:secrecy} and~\eqref{eq:fano}) we have that for any length $n$ code with rate $R$:
\begin{align}
&nR= H(\rvm) \le H(\rvm|\rvby_1) + I(\rvm;\rvby_1)\notag\\
&\le I(\rvm;\rvby_1) + n\eps_n \label{eq:f1}\\
&\le I(\rvm;\rvby_1)-I(\rvm;\rvbz_k) + 2n\eps_n \label{eq:f2}\\
&\le I(\rvm;\rvby_1,\rvby_2,\ldots, \rvby_{M-1},\rvbz_k) - I(\rvm;\rvbz_k) + 2n\eps_n \notag\\
&\le I(\rvm;\rvby_1,\ldots, \rvby_{M-1}| \rvbz_k) + 2n\eps_n \notag\\
&= \dent(\rvby_1,\ldots, \rvby_{M-1}|\rvbz_k) - \notag\\&\qquad\dent(\rvby_1,\ldots, \rvby_{M-1}|\rvbz_k,\rvm) + 2n\eps_n \notag\\
&\le\dent(\rvby_1,\ldots, \rvby_{M-1}|\rvbz_k) - \notag\\&\qquad\dent(\rvby_1,\ldots, \rvby_{M-1}|\rvbz_k,X,\rvm) + 2n\eps_n \notag\\
&= \dent(\rvby_1,\ldots, \rvby_{M-1}|\rvbz_k) - \notag\\&\qquad\dent(\rvbv_1,\ldots, \rvbv_{M-1}) + 2n\eps_n \label{eq:f3}\\
&= \dent(\rvby_1,\ldots, \rvby_{M-1},\rvbz_k) - \dent(\rvbz_k) - \notag \\ &\qquad\frac{n(M-1)}{2}\log (2\pi e) + 2n\eps_n\label{eq:f4}\\
&= \dent(\rvby_1,\ldots, \rvby_M, \rvbz_k)  - h(\rvby_M|\rvby_1,\ldots,\rvby_{M-1},\rvbz_k)\notag\\
&\qquad- \dent(\rvbz_k) - \frac{n(M-1)}{2}\log (2\pi e) + 2n\eps_n\notag\\
&= \dent(\rvby_1,\ldots,\rvby_M)  - \dent(\rvbz_k)\notag\\
&~- h(\rvby_M|\rvby_1,\ldots,\rvby_{M-1},\rvbz_k)+\dent(\rvbz_k|\rvby_1,\ldots,\rvby_M) \notag\\
&+ 2n\eps_n - \frac{n(M-1)}{2}\log (2\pi e)\label{eq:f5}
\end{align}
where~\eqref{eq:f1} and~\eqref{eq:f2} are consequences of Fano's Lemma~\eqref{eq:fano} and the secrecy constraint~\eqref{eq:secrecy} respectively,~\eqref{eq:f3} follows
from the fact that the noise variables $(\rvbv_1,\ldots,\rvbv_{M-1})$ in~\eqref{eq:model2} are independent of $(\rvbx,\rvbz_K)$ and finally~\eqref{eq:f4} follows from the fact that the noise variables are i.i.d. $\cN(0,1)$.

To complete the argument it suffices to show that there exist constants $d_{1k}$ and $d_{2k}$, that are independent of $P$ such that 
\begin{align}
\dent(\rvby_M|\rvby_1,\ldots,\rvby_{M-1},\rvbz_k) \label{eq:dep1} &\ge nd_{1k}\\
 \dent(\rvbz_k|\rvby_1,\ldots,\rvby_M) &\le nd_{2k}\label{eq:dep2}.
\end{align}

To establish~\eqref{eq:dep1} we observe that
\begin{align*}
\dent(\rvby_M|\rvby_1,\ldots,\rvby_{M-1},\rvbz_k) &\ge \dent(\rvby_M|\rvby_1,\ldots,\rvby_{M-1},\rvbz_k,X)\\
&=\dent(\rvbv_M|\rvby_1,\ldots,\rvby_{M-1},\rvbz_k,X)\\
&=\frac{n}{2}\log 2\pi e  \defeq nd_{1k}\end{align*}
where the last relation follows from the fact that the noise variable $\rvbv_M$ is independent of all other variables.

To establish~\eqref{eq:dep2} we observe that  the collection of vectors $(\bh_1,\ldots, \bh_{M})$ constitutes a basis for $\mathbb{R}^M$ i.e., we can express
\begin{equation}
\bg_k = H_{k}\bbe_k \label{eq:dependence}
\end{equation}
where $H_{k}= [\bh_1,\bh_2,\ldots,\bh_{M}] \in\mathbb{R}^{M\times M}$  is a matrix  obtained by stacking the channel vectors of the $M$ legitmate receivers, $\bbe_k \in \mathbb{R}^{M}$ is a vector. Hence we have that
\begin{align}
& \dent(\rvbz_k|\rvby_1,\ldots,\rvby_{M})\\
&=\dent\left(\rvbz_k -  \bbe_k^T\left[\begin{array}{c} \rvby_1\\\vdots\\\rvby_{M}\end{array}\right] \Biggm| \rvby_1,\ldots,\rvby_{M}\right) \notag\\
&\le \dent\left(\rvbz_k - \bbe_k^T\left[\begin{array}{c} \rvby_1\\\vdots\\\rvby_{M}\end{array}\right] \right)\notag\\
&=\dent\left(\rvbw_k -\bbe_k^T\left[\begin{array}{c} \rvbv_1\\\vdots\\\rvbv_{M}\end{array}\right]\right)\label{eq:nc1}\\
&=\frac{n}{2}\log 2 \pi e \left(1 + ||\bbe_k||^2 \right) \defeq nd_{2k} \label{eq:nc2}
\end{align}
where we use~\eqref{eq:dependence} in~\eqref{eq:nc1} and the last relation follows from the fact that the noise vectors consist of independent Gaussian entries $\cN(0,1)$. 

Thus we have from~\eqref{eq:f5} that
\begin{equation}
nR \le \dent(\rvby_1,\ldots,\rvby_M)-\dent(\rvbz_k) + nc_k,
\end{equation}
where
\begin{equation}
c_k = 2\eps_n +\frac{M}{2}\log 2\pi e + \frac{1}{2}\log 2\pi e (1+||\bbe_k||^2)
\end{equation}
is a constant that does not depend on $P$.

\section{Proof of \eqref{eq:LB}}
\label{app:LB}

Define the channel matrices
\begin{equation}
H = \left[\begin{array}{c}\bh_1^T \\\vdots\\\bh_M^T\end{array}\right],\qquad G = \left[\begin{array}{c}\bg_1^T \\\vdots\\\bg_M^T\end{array}\right],
\end{equation}
and the noise matrices
\begin{equation}
\rvW = \left[\begin{array}{c}\rvbw_1 \\\vdots\\\rvbw_M\end{array}\right],\qquad \rvV = \left[\begin{array}{c}\rvbv_1 \\\vdots\\\rvbv_M\end{array}\right],
\end{equation}
so that we can express
\begin{equation}
\left[\begin{array}{c}\rvby_1 \\\vdots\\\rvby_M\end{array}\right] = HX+V,\qquad \left[\begin{array}{c}\rvbz_1 \\\vdots\\\rvbz_M\end{array}\right] = GX+W.\label{eq:channelVectors}
\end{equation}
Then~\eqref{eq:LB} is equivalent to showing that,
\begin{equation}
\dent(GX+W) \ge \dent(HX+V) -nMd,\label{eq:LB-mod}
\end{equation}
for some constant $d$ that does not depend on $P$.  
\begin{align}
&\dent(GX+W) \notag\\
&= \dent(HX+V) - \dent(HX+V|GX+W) \notag\\&\qquad + \dent(GX+W|HX+V) \notag\\
&\ge \dent(HX+V) -\dent(HX+V|GX+W) \notag\\&\qquad+ \dent(GX+W|HX+V,X)\label{eq:nex1}\\
&=\dent(HX+V) - \dent(HX+V|GX+W) + \dent(W) \label{eq:nex2}\\
&=\dent(HX+V) - \dent(HX+V|GX+W) + \frac{n}{2}\log 2\pi e \label{eq:nex3}
\end{align}
where in~\eqref{eq:nex1} we use the fact that conditioning reduces the differential entropy, in~\eqref{eq:nex2} and~\eqref{eq:nex3}, we use the fact that  noise variables in $W$ are independent of $V$ and i.i.d.\ $\cN(0,1)$.

Further note that\begin{align}
&\dent(HX+V|GX+W) \notag\\ &= \dent(V-HG^{-1}W|GX+W)\label{eq:nex4}\\
&\le \dent(V-HG^{-1}W)\notag\\
&=\frac{n}{2}\log2\pi e\det\left(I + HG^{-1}G^{-T}H^T\right)\label{eq:nex5},
\end{align}
where we use the fact that the channel matrices $H$ and $G$ are full rank and invertible in~\eqref{eq:nex4} and that they have i.i.d. $\cN(0,1)$ entries in~\eqref{eq:nex5}. Substituting in~\eqref{eq:nex3} it follows that
\begin{align}
\dent(GX+W) &\ge \dent(HX+V)\notag\\
&\qquad -\frac{n}{2}\log \det\left(I + HG^{-1}G^{-T}H^T\right) 
\end{align}
and thus we can select 
$$d =\frac{1}{2M}\log\det\left(1 + HG^{-1}G^{-T}H^T\right).$$
in~\eqref{eq:LB}.

\section{Bound on $\dent(\rvby_j)$}
\label{app:EntBnd}
Define the input covariance at time $t$, $K_\rvx(t) = E[\rvbx(t)\rvbx(t)^T]$  and let $P_t = \mrm{trace}(K_\rvx(t))$. Recall from the power constraint that $\frac{1}{n}\sum_{t=1}^n P_t \le P $.
Since
$$\rvy_j(t) = \bh_j^T \rvbx(t) + \rvv_j(t)$$
we have that
\begin{align}
\var(\rvy_j(t))&= \bh_j^T K_\rvx(t) \bh_j + 1 \label{eq:o1}\\
&\le \lambda_\mrm{max}(K_\rvx(t))||\bh_j||^2 + 1\label{eq:o2}\\
&\le P_t||\bh_j||^2 + 1\label{eq:o3}
\end{align}where~\eqref{eq:o1} follows from the fact the noise variable $\rvv_j$ is indpendent of the input $\rvbx(t)$,~\eqref{eq:o2} from the variational characterization of eigen values (see e.g.,~\cite{golubVanLoan}) and finally~\eqref{eq:o3} follows from the fact that $P_t \ge \mrm{trace}(K_\rvx(t))$ exceeds the sum of the eigen values and hence exceeds the largest eigen value.  
\begin{align}
\dent(\rvby_j)&=\sum_{t=1}^n \dent(\rvby_j(t))\label{eq:oa1}\\
&=\sum_{t=1}^n \frac{1}{2}\log 2\pi e \left(P_t||\bh_j||^2+1\right)\label{eq:oa2}\\
&\le \frac{n}{2}\log 2\pi e\left(||\bh_j||^2\frac{1}{n}\sum_{t=1}^nP_t +1 \right)\label{eq:oa3}\\
&\le \frac{n}{2}\log 2\pi e\left(||\bh_j||^2P+1 \right)\label{eq:oa4}\\
&\le \frac{n}{2}\max\left(\log P,0\right) + \frac{n}{2}\log 2\pi e(1+||\bh_j||^2) 
\end{align}where~\eqref{eq:oa1} follows from the fact that conditioning reduces the differential entropy,~\eqref{eq:oa2} from the fact that a Gaussian random variable maximizes the differential entropy among all random variables with a fixed variance and~\eqref{eq:oa3} follows from Jensen's inequality, since the $\log(\cdot)$ function is concave and~\eqref{eq:oa4} follows from the power constraint. To get an upper bound on all $1\le j \le M$ we have that
\begin{equation}
\dent(\rvby_j) \le \frac{n}{2}\max\left(\log P,0\right) + nd_1\label{eq:EntBnd}
\end{equation} 
where
\begin{equation}
d_1 = \frac{1}{2}\log 2\pi e(1+\max_{1\le j \le M}||\bh_j||^2)
\end{equation}

\section{Proof of Prop.~\ref{prop:PB_achiev}}
\label{app:PB_achiev}

For any joint distribution of the form\begin{multline}p_{\rvu_1,\rvu_2,\rvx,\rvy_1,\ldots, \rvy_{J_1},\rvz_1,\ldots, \rvz_{J_2}} \\= p_{\rvu_1}p_{\rvu_2}p_{\rvx|\rvu_1,\rvu_2} p_{\rvy_1,\ldots,\rvy_{J_1},\rvz_1,\ldots,\rvz_{J_2}|\rvx}.\end{multline} we show that the following rate pair is achievable
\begin{align}
R_1 &= \min_{j} I(\rvu_1;\rvy_j) - I(\rvu_1;\rvu_2,\rvz_k^\star) - \delta\label{eq:R1a}\\
R_2 &= \min_{k} I(\rvu_2;\rvz_k) - I(\rvu_2;\rvu_1,\rvy_j^\star) - \delta,\label{eq:R2a}
\end{align}where $\delta > 0$ is an arbitrary constant and where we have introduced, $$ k^\star = \arg\max_k I(\rvu_1;\rvu_2,\rvz_k),\quad j^\star = \arg \max_j I(\rvu_2;\rvu_1,\rvy_j).$$

We start by constructing codebooks
\begin{equation}
\begin{aligned}
\cC_1 = \left\{u_{1ab}^n: a=1,\ldots, 2^{nR_1}, b=1,\ldots, 2^{n I(\rvu_1;\rvu_2,\rvz_k^\star)}\right\}\\
\cC_2 = \left\{u_{2ab}^n: a=1,\ldots, 2^{nR_2}, b=1,\ldots, 2^{n I(\rvu_2;\rvu_1,\rvy_j^\star)}\right\}\end{aligned}\label{eq:cC}
\end{equation}
We assume that the codewords in $\cC_1$ belong to the set $T_\eps^n(\rvu_1)$ of strongly typical sequences whereas the codewords in $\cC_2$ belong to the set $T_\eps^n(\rvu_2)$. Given messages $\rvm_1$ and $\rvm_2$ the encoder sets $a_1 = \rvm_1$, $a_2 = \rvm_2$ and selects 
$b_1$, uniformly from $\{1,\ldots, 2^{nI(\rvu_1;\rvu_2,\rvz_k^\star)}\}$ and $b_2$ uniformly from the set $\{1,\ldots, 2^{nI(\rvu_2;\rvu_1,\rvy_j^\star)}\}$. It selects codewords $(u_{1a_1b_1}^n$ and $u_{2a_2b_2}^n)$ from $\cC_1$ and $\cC_2$ respectively and then transmits $\rvx^n$ generated by passing  codewords through the memoryless, fictitious channel $p_{\rvx|\rvu_1,\rvu_2}(\cdot)$.

With the choice of $R_1$ and $R_2$ specified in~\eqref{eq:R1a} and~\eqref{eq:R2a} it can be shown, (see e.g.,~\cite{LiuMaricSpasojevicYates:07}), that there exist codebooks $\cC_1$ and $\cC_2$ in~\eqref{eq:cC}  such that the error probability at each receiver is smaller than any target value. Furthermore receiver $k^\star$ in group 2 can decode the index $b_1$ with high probability if it is revealed $\rvm_1$ in addition to $(\rvz_{k^\star}^n, \rvu_2^n)$. Similarly  eavesdropper $j^\star$ can decode the index $b_2$ if it is revealed $\rvm_2$ in addition to $(\rvy_{j^\star}^n,\rvu_1^n)$ i.e.,
\begin{align}
\max\!\left\{\frac{1}{n}H(\rvu_1^n|\rvz_{k^\star}^n,\rvu_2^n,\rvm_1),\frac{1}{n}H(\rvu_2^n|\rvy_{j^\star}^n,\rvu_1^n,\rvm_2) \right\}\!\le\!\eps_n \label{eq:fanoAC}
\end{align}For such a codebook, we show that for some sequence $\eps_n$ that approaches zero as $n\rightarrow \infty$, we have that
\begin{equation}\begin{aligned}
\frac{1}{n}I(\rvm_1;\rvz_k^n) &\le \eps_n,\qquad
\frac{1}{n}I(\rvm_2;\rvy_j^n) \le \eps_n \\
\end{aligned}\end{equation}
for $k=1,\ldots, J_2$ and $j=1,\ldots, J_1$. We first observe that it suffices for us to show
that\begin{equation}\begin{aligned}\frac{1}{n}I(\rvm_1;\rvz_{k^\star}^n) &\le \del_n,\qquad
\frac{1}{n}I(\rvm_2;\rvy_{j^\star}^n) \le \del_n. \label{eq:secAc}\\
\end{aligned}\end{equation}
As we show below, one can provide sufficient side information to enhance each eavesdropper so that it is equivalent to the strongest eavesdropper. Then it is easy to verify that if the coding scheme guarantees~\eqref{eq:secAc} then the messages are also secure from the enhanced eavesdroppers and hence the original eavesdroppers. 
\begin{lemma}
For each $k=1,\ldots, J_2$ there exists a variable $\rvzt_k$ that satisfies $(\rvu_1,\rvu_2) \rightarrow \rvx \rightarrow (\rvz_k,\rvz_{k^\star}) \rightarrow \rvzt_k$ such that
$I(\rvu_1;\rvu_2,\rvz_{k^\star}) = I(\rvu_1;\rvu_2,\rvz_k,\rvzt_k)$. Likewise for each $j=1,\ldots, J_1$ there exists a variable $\rvyt_j$ that satisfies $(\rvu_1,\rvu_2) \rightarrow \rvx \rightarrow (\rvy_j,\rvy_{j^\star}) \rightarrow\rvyt_j$ such that $I(\rvu_2;\rvu_1,\rvy_{j^\star})= I(\rvu_2;\rvu_1,\rvy_j,\rvyt_j)$.
\label{lem:enhance}\end{lemma}
The proof of Lemma~\ref{lem:enhance} will be provided at the end of this section.  To establish~\eqref{eq:secAc} consider:
\begin{align}
&\frac{1}{n}H(\rvm_1|\rvz_{k^\star}^n) \ge\frac{1}{n}H(\rvm_1|\rvz_{k^\star}^n,\rvu_2^n)\label{eq:condRed}\\
&\ge \frac{1}{n} H(\rvu_1^n,\rvm_1|\rvz_{k^\star}^n,\rvu_2^n) - \del_n \label{eq:fano1}\\
&=\frac{1}{n} H(\rvu_1^n|\rvz_{k^\star}^n,\rvu_2^n) - \del_n \label{eq:detmsg}\\
&= \frac{1}{n}H(\rvu_1^n|\rvu_2^n) - \frac{1}{n}I(\rvu_1^n;\rvz_{k^\star}^n|\rvu_2^n) -\del_n 
\label{eq:step1}\end{align}
where~\eqref{eq:condRed} follows from the fact that conditioning on $\rvu_2^n$ reduces the entropy, \eqref{eq:fano1} follows from via~\eqref{eq:fanoAC},~\eqref{eq:detmsg} follows from the fact that $\rvm_1$ is a deterministic function of $\rvu_1^n$ and hence can be dropped from the conditioning. We separately bound the two terms in~\eqref{eq:step1}.
\begin{align}
\frac{1}{n}H(\rvu_1^n|\rvu_2^n) =\frac{1}{n}H(\rvu_1^n) = \min_{j} I(\rvu_1;\rvy_{j^\star}) - \delta, \label{eq:t1}
\end{align} where we have used the fact that the codeword $\rvu_1^n$ is independently selected of $\rvu_2^n$ in our construction and is uniformly distributed over the set $\cC_1$. We upper bound the second term in~\eqref{eq:step1} as follows. Since the cascade channel $(\rvu_1,\rvu_2)\rightarrow \rvz_{k^\star}$ is memoryless, it an be easily verified that
\begin{align}
I(\rvu_1^n;\rvz_{k^\star}^n|\rvu_2^n) &\le \sum_{i=1}^n I(\rvu_{1i};\rvz_{k^\star i}|\rvu_{2i})
\end{align}
Furthermore for each codebook in the ensemble of typical codebooks we have from the weak law of large numbers that the summation on the right hand side converges to $n I(\rvu_1;\rvz_{k^\star}|\rvu_2)$ i.e., the mutual information evaluated with the original input distribution.  Thus we can write
\begin{align}
I(\rvu_1^n;\rvz_{k^\star}^n|\rvu_2^n) &\le  n I(\rvu_{1};\rvz_{k^\star }|\rvu_{2}) - n o_n(1),\label{eq:t2}
\end{align}where $o_n(1)$ converges to zero as $n\rightarrow\infty$.
Substituting~\eqref{eq:t1} and~\eqref{eq:t2} into~\eqref{eq:step1} we establish the first half of~\eqref{eq:secAc}. The second half of~\eqref{eq:secAc} can be established in an analogous manner. 

It remains to provide a proof of Lemma~\ref{lem:enhance} which we do below.

\begin{proof}[Lemma~\ref{lem:enhance}]
The construction of random variable $\rvzt_k$ follows the same approach as in the case of compound wiretap channel~\cite{liangKramer:08}. In particular suppose that $I_k \in\{k,k^\star\}$ is a random variable independent of all other variables. Let $\Pr(I_k=k^\star)=p$.

Define a new random variable, $\rvz_{I_k} = (\rvz_{I_k},I_k)$ and consider the function $f(p) = I(\rvu_1;\rvu_2,\rvz_k,\rvz_{I_k})$. It is clear that $f(0) = I(\rvu_1;\rvu_2,\rvz_k) \le I(\rvu_1;\rvu_2,\rvz_{k^\star})$ whereas $f(1) = I(\rvu_1;\rvu_2,\rvz_k,\rvz_{k^\star}) \ge I(\rvu_1;\rvu_2,\rvz_{k^\star})$. Thus there exists a value of $p^\star \in [0,1]$ such that $$I(\rvu_1;\rvu_2,\rvz_k,\rvz_{I_k}) = I(\rvu_1;\rvu_2,\rvz_{k^\star})$$ and the resulting random variable in Lemma~\ref{lem:enhance} is given by $\rvzt_k = (\rvz_{I_k},I_k)$. The construction of $\rvyt_k$ follows in an analogous manner. 
\end{proof}

\section*{Acknowledgements}
Insightful discussions with  Frank Kschischang and Suhas Diggavi during the early stages of this work are gratefully acknowledged. The author also thanks T.~Liu for careful reading of the proof in Theorem~\ref{thm:CWC:UB} in an earlier preprint of this work.

\bibliographystyle{IEEEtran}

\begin{thebibliography}{10}
\providecommand{\url}[1]{#1}
\csname url@samestyle\endcsname
\providecommand{\newblock}{\relax}
\providecommand{\bibinfo}[2]{#2}
\providecommand{\BIBentrySTDinterwordspacing}{\spaceskip=0pt\relax}
\providecommand{\BIBentryALTinterwordstretchfactor}{4}
\providecommand{\BIBentryALTinterwordspacing}{\spaceskip=\fontdimen2\font plus
\BIBentryALTinterwordstretchfactor\fontdimen3\font minus
  \fontdimen4\font\relax}
\providecommand{\BIBforeignlanguage}[2]{{%
\expandafter\ifx\csname l@#1\endcsname\relax
\typeout{** WARNING: IEEEtran.bst: No hyphenation pattern has been}%
\typeout{** loaded for the language `#1'. Using the pattern for}%
\typeout{** the default language instead.}%
\else
\language=\csname l@#1\endcsname
\fi
#2}}
\providecommand{\BIBdecl}{\relax}
\BIBdecl

\bibitem{khistiITA10}
\BIBentryALTinterwordspacing
A.~Khisti, ``On the compound {MISO} wiretap channel,'' in \emph{Information
  Theory and its Applications Worksop}, Jan. 2010. [Online]. Available:
  \url{http://ita.ucsd.edu/workshop/10/files/paper/paper\_1231.pdf}
\BIBentrySTDinterwordspacing

\bibitem{wyner:75Wiretap}
A.~D. Wyner, ``The wiretap channel,'' \emph{Bell Syst.\ Tech.\ J.}, vol.~54,
  pp. 1355--87, 1975.

\bibitem{khistiTchamWornell:07}
A.~Khisti, A.~Tchamkerten, and G.~W. Wornell, ``Secure {B}roadcasting over
  fading channels,'' \emph{IEEE Trans.\ Inform.\ Theory, Special Issue on
  Information Theoretic Security}, 2008.

\bibitem{gopalaLaiElGamal:06Secrecy}
P.~Gopala, L.~Lai, and H.~E. {G}amal, ``On the secrecy capacity of fading
  channels,'' \emph{IEEE Trans.\ Inform.\ Theory}, Oct, 2008.

\bibitem{liangPoor07}
Y.~Liang, H.~V. Poor, and S.~Shamai, ``Secure communication over fading
  channels,'' \emph{IEEE Trans.\ Inform.\ Theory}, June, 2008.

\bibitem{khistiWornellEldar:07}
A.~Khisti, G.~W. Wornell, A.~Wiesel, and Y.~Eldar, ``On the {G}aussian {MIMO}
  wiretap channel,'' in \emph{Proc.\ Int.\ Symp.\ Inform.\ Theory}, Nice, 2007.

\bibitem{khistiWornell:09a}
A.~Khisti and G.~W. Wornell, ``Secure transmission with multiple antennas: The
  {MISOME} wiretap channel,'' \emph{To Appear, IEEE Trans.\ Inform.\ Theory}.

\bibitem{khistiWornell:09b}
------, ``Secure transmission with multiple antennas: The {MIMOME} wiretap
  channel,'' \emph{To Appear, IEEE Trans.\ Inform.\ Theory}.

\bibitem{LyLiu:09}
H.~D. Ly, T.~Liu, and Y.~Liang, ``Multiple-input multiple-output gaussian
  broadcast channels with common and confidential messages,'' \emph{IEEE
  Trans.\ Inform.\ Theory}, July, 2009, submitted.

\bibitem{LiuLiu:09}
R.~Liu, T.~Liu, H.~Poor, and S.~Shamai, ``Multiple-input multiple-output
  gaussian broadcast channels with confidential messages,'' \emph{IEEE Trans.\
  Inform.\ Theory}, To Appear.

\bibitem{LiuShamai:09}
T.~Liu and S.~Shamai, ``A note on the secrecy capacity of the multiple-antenna
  wiretap channel,'' \emph{IEEE Trans.\ Inform.\ Theory}, June, 2009.

\bibitem{shafieeLiuUlukus:07}
S.~Shafiee, N.~Liu, and S.~Ulukus, ``Towards the secrecy capacity of the
  {G}aussian {MIMO} wire-tap channel: The 2-2-1 channel,'' \emph{IEEE Trans.\
  Inform.\ Theory}, Sept., 2009.

\bibitem{EkremUlukus:07}
E.~Ekrem and S.~Ulukus, ``The secrecy capacity region of the gaussian mimo
  multi-receiver wiretap channel,'' \emph{IEEE Trans.\ Inform.\ Theory}, March,
  2009, submitted.

\bibitem{Ghadamali}
G.~Bagherikaram, A.~S. Motahari, and A.~K. Khandani, ``Secrecy capacity region
  of gaussian broadcast channel,'' in \emph{CISS}, 2009.

\bibitem{liangKramer:08}
Y.~Liang, G.~Kramer, H.~V. Poor, and S.~S. (Shitz), ``Compound wire-tap
  channels,'' \emph{EURASIP Journal on Wireless Communications and Networking,
  Special Issue on Wireless Physical Layer Security}, submitted 2008.

\bibitem{Liuprabhakaran:07}
T.~Liu, V.~Prabhakaran, and S.~Vishwanath, ``The secrecy capacity of a class of
  non-degraded parallel gaussian compound wiretap channels,'' Jul. 2008.

\bibitem{elGamal:80}
A.~A. {El Gamal}, ``Capacity of the product and sum of two un-matched broadcast
  channels,'' \emph{Probl. Information Transmission}, pp. 3--23, 1980.

\bibitem{EkremUlukus:09}
E.~Ekrem and S.~Ulukus, ``Degraded compound multi-receiver wiretap channels,''
  \emph{IEEE Trans.\ Inform.\ Theory}, Oct., 2009, submitted.

\bibitem{PeronDiggavi:09}
E.~Perron, S.~N. Diggavi, and E.~Telatar, ``On cooperative wireless network
  secrecy,'' in \emph{INFOCOM}, April 2009.

\bibitem{realInterf1}
A.~S. Motahari, S.~O. Gharan, and A.~K. Khandani, ``Real interference alignment
  with real numbers,'' \emph{Submitted to IEEE Trans.\ Inform.\ Theory,
  http://arxiv.org/abs/0908.1208}, 2009.

\bibitem{realInterf2}
A.~S. Motahari, S.~O. Gharan, M.~Maddah-Ali, and A.~K. Khandani, ``Real
  interference alignment: Exploiting the potential of single antenna systems,''
  \emph{Submitted to IEEE Trans.\ Inform.\ Theory,
  http://arxiv.org/abs/0908.2282}, 2009.

\bibitem{compoundMIMO1}
M.~Maddah-Ali, ``On the degrees of freedom of the compound mimo broadcast
  channels with finite states,'' \emph{Submitted to IEEE Trans.\ Inform.\
  Theory, http://arxiv.org/abs/0909.5006}, 2009.

\bibitem{compoundMIMO2}
T.~Gou and C.~W. S.~A.~Jafar, ``On the degrees of freedom of finite state
  compound wireless networks - settling a conjecture by weingarten et. al.''

\bibitem{caiLam00}
N.~Cai and K.~Y. Lam, ``How to broadcast privacy: Secret coding for
  deterministic broadcast channels,'' \emph{Numbers, Information, and
  Complexity (Festschrift for Rudolf Ahlswede), eds: I. Alth¨ofer, N. Cai, G.
  Dueck, L. Khachatrian, M. Pinsker, A. Sarkozy, I. Wegener, and Z. Zhang}, pp.
  353--368, 2000.

\bibitem{LiuMaricSpasojevicYates:07}
R.~Liu, I.~Maric, P.~Spasojevic, and R.~D. Yates, ``Discrete memoryless
  interference and broadcast channels with confidential messages: Secrecy
  capacity regions,'' \emph{IEEE Trans.\ Inform.\ Theory}, June 2009.

\bibitem{IA:08}
O.~Koyluoglu, H.~E. Gamal, L.~Lai, and H.~V. Poor, ``Interference alignment for
  secrecy,'' \emph{IEEE Trans.\ Inform.\ Theory}, Oct. 2008, submitted.

\bibitem{heYener:09}
X.~He and A.~Yener, ``Secure degrees of freedom for gaussian channels with
  interference: Structured codes outperform gaussian signaling,'' \emph{IEEE
  Trans.\ Inform.\ Theory}.

\bibitem{jafar:10}
C.~Wang, T.~Guo, and S.~A. Jafar, ``Aiming perfectly in the dark - blind
  interference alignment through staggered antenna switching,'' \emph{e-print
  arXiv:1002.2720}.

\bibitem{golubVanLoan}
G.~Golub and C.~F.~V. Loan, \emph{Matrix Computations (3rd ed)}.\hskip 1em plus
  0.5em minus 0.4em\relax Johns Hopkins University Press, 1996.

\end{thebibliography}

\end{document}